\documentclass[lettersize,journal]{IEEEtran}
\IEEEoverridecommandlockouts

\usepackage{amssymb}
\usepackage[cmex10]{amsmath}
\usepackage{stfloats}
\usepackage{graphicx}
\usepackage{subfigure}
\usepackage{tabularx}
\usepackage{epsfig,epsf,color,balance,cite}
\usepackage{verbatim}
\usepackage{url}
\usepackage{bm}
\usepackage{booktabs}
\usepackage{siunitx}

\usepackage{amsmath}
\allowdisplaybreaks[2]
\usepackage{ntheorem}
\newtheorem{theorem}{Theorem}
\newtheorem{lemma}{Lemma}

\usepackage{algorithm}
\usepackage{algorithmic}

\usepackage{color}
\definecolor{myc1}{rgb}{0,0,0}

\begin{document}

\title{Agentic AI-Empowered Wireless Agent Networks With Semantic-Aware Collaboration via ILAC}

\author{Zhouxiang Zhao, 
        Jiaxiang Wang,
        Zhaohui Yang, 
        Kun Yang,
        Zhaoyang Zhang,~\IEEEmembership{Senior Member,~IEEE,}\\
        Mingzhe Chen,~\IEEEmembership{Senior Member,~IEEE,}
        and Kaibin Huang,~\IEEEmembership{Fellow,~IEEE}

\thanks{Zhouxiang Zhao, Zhaohui Yang, Kun Yang, and Zhaoyang Zhang are with the College of Information Science and Electronic Engineering, Zhejiang University, and also with Zhejiang Provincial Key Laboratory of Info. Proc., Commun. \& Netw. (IPCAN), Hangzhou 310027, China (e-mails: \{zhouxiangzhao, yang\_zhaohui, kunyang20, ning\_ming\}@zju.edu.cn).}
\thanks{Jiaxiang Wang is with the Department of Engineering, King's College London, London, WC2R 2LS, UK (e-mail: jiaxiang.wang@kcl.ac.uk).}
\thanks{Mingzhe Chen is with Department of Electrical and Computer Engineering and Institute for Data Science and Computing, University of Miami, Coral Gables, FL 33146, USA (e-mail: mingzhe.chen@miami.edu).}
\thanks{Kaibin Huang is with the Department of Electrical and Electronic Engineering, The University of Hong Kong, Hong Kong SAR, China (e-mail: huangkb@eee.hku.hk).}
}

\maketitle

\begin{abstract}
The rapid development of agentic artificial intelligence (AI) is driving future wireless networks to evolve from passive data pipes into intelligent collaborative ecosystems under the emerging paradigm of integrated learning and communication (ILAC). However, realizing efficient agentic collaboration faces challenges not only in handling semantic redundancy but also in the lack of an integrated mechanism for communication, computation, and control. To address this, we propose a wireless agent network (WAN) framework that orchestrates a progressive knowledge aggregation mechanism.
Specifically, we formulate the aggregation process as a joint energy minimization problem where the agents perform semantic compression to eliminate redundancy, optimize transmission power to deliver semantic payloads, and adjust physical trajectories to proactively enhance channel qualities.
To solve this problem, we develop a hierarchical algorithm that integrates inner-level resource optimization with outer-level topology evolution.
Theoretically, we reveal that incorporating a potential field into the topology evolution effectively overcomes the short-sightedness of greedy matching, providing a mathematically rigorous heuristic for long-term energy minimization.
Simulation results demonstrate that the proposed framework achieves superior energy efficiency and scalability compared to conventional benchmarks, validating the efficacy of semantic-aware collaboration in dynamic environments.
\end{abstract}


\IEEEpeerreviewmaketitle

\section{Introduction}
\IEEEPARstart{T}{he} evolution of sixth-generation (6G) networks is witnessing a paradigm shift from connecting humans and devices to connecting autonomous intelligences, specifically embodied artificial intelligence (AI) agents \cite{11339915,11370176,11152698}. 
Driven by the recent proliferation of foundation models, the agents, equipped with sensing, reasoning, and actuating capabilities, are transforming into agentic AI, capable of autonomously perceiving complex environments and executing collaborative tasks without human intervention \cite{11373363,11303197}. 
This transformation has given rise to the concept of Internet of Agents (IoA) \cite{11207716}, where the primary objective evolves from simple data exchange to high-level knowledge sharing and collaborative problem-solving. 
To effectively support such advanced agentic collaboration, it is imperative to embrace the paradigm of integrated learning and communication (ILAC) \cite{11381448}, which seamlessly couples AI-driven cognitive inference with wireless transmission.

In mission-critical scenarios such as large-scale surveillance, disaster response, and battlefield reconnaissance, a swarm of mobile agents is often deployed to cover a wide geographic area \cite{10599391}. 
Conventionally, these systems operate on a ``sense-and-forward" basis, where the agents capture raw sensory data (e.g., high-resolution images or video streams) and transmit the packets to a central sink for processing. 
However, this traditional paradigm faces severe scalability challenges. 
The transmission of massive raw data imposes prohibitive demands on the limited wireless spectrum and rapidly depletes the onboard battery energy of mobile agents. 
More critically, raw data streams often contain significant semantic redundancy. For instance, multiple agents may capture overlapping views of the same object, or a single agent may transmit static background frames that convey negligible information gain for the overall situational awareness \cite{11395598}.

To address these inefficiencies, semantic communication has emerged as a promising solution, shifting the focus from the accurate transmission of bits to the effective conveyance of meaning \cite{gunduz2022beyond,10915662}. 
By leveraging the generative and inferential capabilities of embodied large models (ELMs), the agents can now locally process raw observations into compact semantic representations \cite{10720863}. 
Furthermore, the integration of retrieval-augmented generation (RAG) allows the agents to synthesize new observations with existing knowledge bases, effectively compressing information based on contextual relevance before transmission \cite{11303308,11297177}. 
This capability is pivotal for realizing what we term a wireless agent network (WAN), where the network load is determined not by the volume of pixels, but by the density of valuable knowledge.

Despite these advancements, realizing an energy-efficient WAN for collaboration remains a challenge. 
Unlike static sensor networks or point-to-point semantic links, a mobile agent network involves a complex coupling of physical mobility, computational inference, and wireless transmission \cite{11370843}. 
The agents must autonomously decide where to move to optimize channel conditions, how much computation resource to invest in semantic compression, and how to dynamically form network topologies to aggregate dispersed information. 
Existing works on multi-agent systems often treat communication, computation, and control as separate optimization problems or rely on fixed topologies that fail to exploit the progressive knowledge aggregation effect, where information becomes increasingly refined and compact as it propagates through the network hierarchy \cite{11373008}. 
Therefore, there is an urgent need for a holistic framework that jointly orchestrates agent mobility, ELM-based semantic processing, and dynamic collaborative networking to minimize system-wide energy consumption while ensuring rigorous latency and coverage constraints.

\subsection{Related Works}
This section reviews the state-of-the-art across three pivotal domains essential to the proposed WAN: mobile collaborative sensing and computing, semantic communication and generative AI, agentic AI and IoA.

\subsubsection{Mobile Collaborative Sensing and Computing}
Mobile collaborative sensing, encompassing paradigms such as mobile edge computing (MEC) \cite{8016573,11159303} and wireless sensor network (WSN) \cite{6747349}, has been extensively studied as a means to extend perception capabilities beyond individual limits. 
Early works focused on optimizing the trajectories of uncrewed aerial vehicles (UAVs) to maximize coverage or minimize energy consumption while satisfying connectivity constraints \cite{9354739,9454157}. 
For instance, extensive research has been dedicated to clustering algorithms and data aggregation trees to reduce redundant transmissions \cite{7112038,10054625}. 
However, these traditional approaches typically treat data as bitstreams, focusing on metrics such as throughput and packet error rate. 
In the context of mission-critical surveillance where the agents capture high-dimensional sensory data, such bit-centric methodologies often lead to excessive bandwidth consumption and latency, as they fail to exploit the inherent semantic sparsity of the information. 

\subsubsection{Semantic Communication and Generative AI}
The paradigm shift from syntax-oriented to goal-oriented communication has garnered significant attention \cite{11220256}. 
Semantic communication systems utilize deep learning models to extract and transmit only the essential meaning of source data, significantly reducing the transmission burdens \cite{11006980,9955312}. 
Recently, the integration of generative AI and large language models has further empowered this domain \cite{10841377}. 
Generative AI-aided frameworks enable the receivers to reconstruct high-fidelity data from highly compressed semantic prompts by leveraging shared knowledge bases \cite{10679152}. 
Despite these advancements, the majority of existing semantic communication literature focuses on optimizing point-to-point links \cite{11008477,ZHAO2024107055}. 
The research on networked semantic aggregation, where the information from multiple sources must be fused, de-duplicated, and evolved hierarchically, remains nascent. 

\subsubsection{Agentic AI and IoA}
The evolution of 6G is propelling the transition from connected devices to connected intelligence, giving rise to the IoA \cite{10648594,10504634}. 
Recent surveys have outlined the architectural vision of IoA, emphasizing the need for protocols that support intent-based interaction and collective intelligence \cite{11207716,chen2024internet}. 
However, existing studies on agentic AI and IoA often focus on high-level multi-agent coordination strategies while abstracting away the underlying physical network constraints \cite{li2025llm,tang2025enhanced,11298134}. 
There is a lack of rigorous mathematical frameworks that capture the intricate coupling between the agents' physical mobility, their cognitive inference processes, and their semantic interactions. 

\subsection{Contributions}
To bridge the gap between high-level agentic collaboration and underlying physical network constraints, we propose the WAN framework. 
Distinct from conventional bit-centric forwarding and static point-to-point semantic links, we introduce a progressive knowledge aggregation mechanism. 
This paradigm orchestrates network-level semantic fusion and deduplication, enabling the information to evolve from fragmented raw data into compact, high-value knowledge as it propagates through the network. 
Furthermore, we establish a rigorous mathematical formulation that seamlessly couples the agents' physical mobility, ELM-based cognitive inference, and wireless transmission. 
This holistic modeling addresses the critical lack of physical-layer grounding in existing IoA studies, transforming abstract agentic capabilities into physically realizable network operations. 
To the best of our knowledge, this is the first work to systematically investigate the agentic AI-empowered WAN. The main contributions of this paper are summarized as follows:
\begin{itemize}
    \item We propose a novel WAN framework that transcends traditional bit-centric data forwarding by introducing a progressive knowledge aggregation mechanism. By integrating ELMs with RAG, the framework empowers agents to progressively fuse, deduplicate, and semantically compress distributed sensory data as it propagates through a dynamically evolving topology. This paradigm shift ensures that network load is determined by knowledge density rather than raw pixel volume, significantly reducing redundant traffic.
    \item We establish a unified mathematical framework that deeply couples the agents' physical mobility, semantic inference, and wireless transmission. Unlike existing multi-agent coordination studies that abstract away physical-layer realities, our holistic approach explicitly captures the intricate trade-offs among trajectory adjustments, compression operations, and transmit power allocations, ensuring that the high-level semantic collaboration is strictly grounded by physical latency deadlines and geometric security coverage constraints.
    \item To resolve the intractable coupling between continuous resource allocation and discrete network formation, we design a hierarchical matching and adaptive planning (H-MAP) framework. The core algorithmic novelty lies in our predictive topology evolution strategy, which incorporates a physical potential field into the graph matching process. This theoretically grounded design anticipates future aggregation trajectories and costs, effectively overcoming the severe short-sightedness of conventional greedy matching methods and ensuring scalable, long-term network energy minimization.
\end{itemize}

The remainder of this paper is organized as follows. 
Section \ref{Sec:sm} introduces the system model of the proposed WAN. 
Section \ref{sec:problem_formulation} mathematically formulates the network energy minimization problem as a mixed-integer non-linear programming (MINLP) task. 
To address the computational intractability, Section \ref{Sec:ad} proposes the H-MAP algorithm, which decouples the problem into inner-level joint motion-resource optimization and outer-level dynamic topology planning. 
Section \ref{Sec:sr} presents extensive simulation results to validate the effectiveness and scalability of the proposed scheme. 
Finally, Section \ref{Sec:c} concludes the paper.

\section{System Model}\label{Sec:sm}

\begin{figure}[t]
    \centering
    \includegraphics[width=\linewidth]{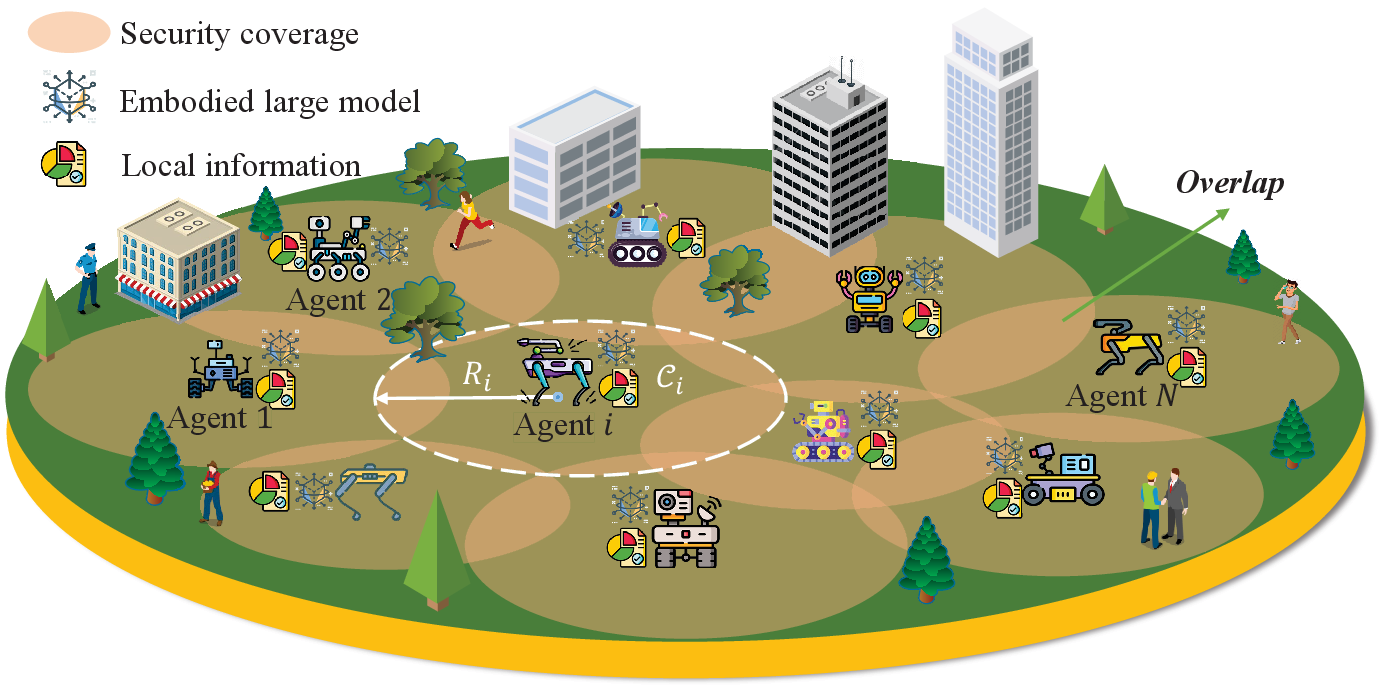}
    \caption{An illustration of the considered  wireless AI agent-assisted surveillance scenario.}
    \label{fig.scenario}
\end{figure}

\subsection{General Setup}
\subsubsection{Scenario Description}
Consider a mission-critical surveillance scenario deployed over a large-scale geographic area, served by a set of $N$ autonomous mobile agents, denoted by $\mathcal{N} = \{1, 2, \dots, N\}$. To ensure continuous and ubiquitous security coverage, each agent $i \in \mathcal{N}$ is mandated to patrol within a specific circular region $\mathcal{C}_i$, as shown in Fig.~\ref{fig.scenario}.

The primary objective of the considered WAN is to aggregate the distributed, local sensing data from all agents into a comprehensive global situation report at a sink node. Unlike conventional sensor networks that forward raw data packets, the agents in our system are equipped with ELMs. This capability enables them to perform semantic processing, leveraging RAG to fuse incoming information with local observations, thereby eliminating redundancy and synthesizing high-level artifacts.

\begin{figure}[t]
    \centering
    \includegraphics[width=\linewidth]{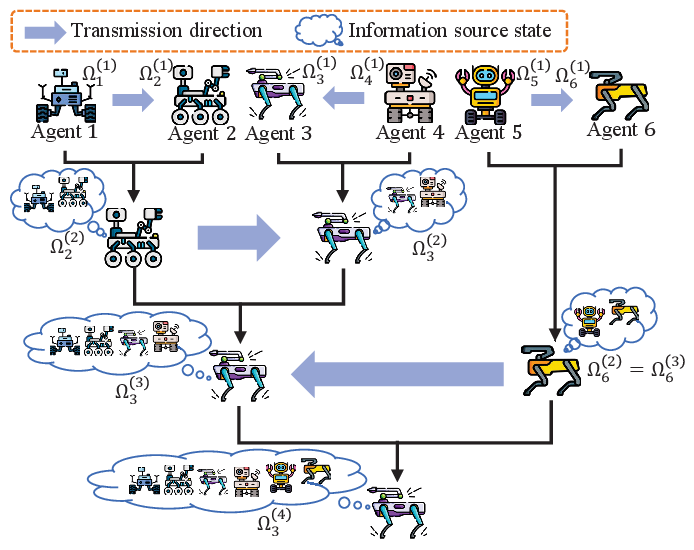}
    \caption{An example of the progressive knowledge aggregation mechanism: the information of six agents is aggregated via pair transmission.}
    \label{fig.tree}
\end{figure}

\subsubsection{Workflow}
The information aggregation process proceeds in discrete time rounds, indexed by $k = 1, 2, \dots, K$. The workflow is characterized by a progressive knowledge aggregation mechanism, where the system evolves from a fully distributed state to a centralized state, as illustrated in Fig.~\ref{fig.tree}. The operation in each round $k$ involves two coupled levels of decision-making:
\begin{itemize}
    \item \textbf{Outer-Level Topology Evolution:} The system dynamically determines a pairing strategy based on the current spatial distribution and states of the agents. This topology must satisfy physical connectivity constraints while aiming to minimize network-wide energy consumption.
    \item \textbf{Inner-Level Interaction Protocol:} Once a pair of agents $(i, j)$ is established, they execute a joint move-compute-communicate protocol. Specifically, the agents autonomously optimize their physical positions within their respective patrol regions to improve channel conditions, utilize local ELMs to compress data based on semantic correlation, and transmit the synthesized payload.
\end{itemize}


\subsection{Network Geometry and Information State}\label{subsec:network_geometry}
\subsubsection{Geometric Constraints and Patrol Regions}
The surveillance area is partitioned into $N$ distinct, potentially overlapping sub-regions, each assigned to a specific agent for continuous monitoring. While agents are endowed with mobility to optimize communication channels, their movement is strictly confined to maintain service availability within their respective zones. Consequently, the instantaneous position of agent $i$, denoted by $\mathbf{p}_i \in \mathbb{R}^2$, is subject to the following hard geometric constraint:
\begin{equation}
    \| \mathbf{p}_i - \mathbf{c}_i \|_2 \le R_i, \quad \forall i \in \mathcal{N},
    \label{eq:geo_constraint}
\end{equation}
where $\mathbf{c}_i \in \mathbb{R}^2$ represents the center coordinate, $R_i$ is the coverage radius, and $\| \cdot \|_2$ denotes the Euclidean norm.

\subsubsection{Information Source State Evolution}
The core objective of the system is to aggregate spatially distributed sensing data into a global report. Unlike conventional packet-forwarding networks, our system emphasizes the provenance of information to facilitate semantic deduplication.

We define the information source state $\Omega_i^{(k)} \subseteq \mathcal{N}$ as the set of unique agent indices whose original sensing data is contained within the memory of agent $i$ at the beginning of round $k$. The evolution of this state follows a set-theoretic union rule. 
\begin{itemize}
    \item \textbf{Initialization:} At the initial round $k=1$, each agent possesses only its local sensing data. Thus, the state is initialized as $\Omega_i^{(1)} = \{i\}, \forall i \in \mathcal{N}$.
    \item \textbf{Recursive Update:} During round $k$, if a directed communication link is established from agent $i$ (sender) to agent $j$ (receiver), agent $j$ acquires not only the current data payload of agent $i$ but also the semantic context of all underlying information sources accumulated by agent $i$. Consequently, the information source state of the receiver updates as $\Omega_j^{(k+1)} = \Omega_j^{(k)} \cup \Omega_i^{(k)}$.
\end{itemize}
Upon successful transmission, the sender $i$ ceases participation in subsequent aggregation rounds, effectively pruning the active topology.

\subsubsection{Accumulated Data Payload Dynamics}
A critical distinction of our model is the handling of data volume. In semantic communication networks empowered by ELMs, the size of the aggregated data is not a linear summation of raw data packets. Instead, it is a dynamic quantity resulting from successive semantic compression and fusion operations.

Let $L_{i}^{(k)}$ (in bits) denote the actual size of the accumulated data payload carried by agent $i$ at the start of round $k$. The evolution of $L_{i}^{(k)}$ is history-dependent, reflecting the efficacy of compression in all preceding interactions. When agent $j$ receives data from agent $i$ in round $k$, the new payload size $L_{j}^{(k+1)}$ is formulated as:
\begin{equation}
    L_{j}^{(k+1)} = \eta_{j}L_{j}^{(k)} + \eta_{i}L_{i}^{(k)},
    \label{eq:payload_update}
\end{equation}
where $\eta_i$ and $\eta_j$ represents the semantic compression ratio applied to agent $i$'s and agent $j$'s payload during the interaction, respectively. 


\subsection{Mobility Model}\label{subsec:mobility_model}
\subsubsection{Mobility Dynamics}
In the proposed collaborative framework, agents are capable of autonomous movement to optimize communication channel conditions. Consider a specific interaction round $k$, where agent $i$ plans to move from its current initial position $\mathbf{p}_i^{\text{start}}$ to an optimized target position $\mathbf{p}_i^{\text{end}}$ within its patrol region $\mathcal{C}_i$. We assume the agent traverses the Euclidean path between these points with a constant cruising velocity $v_i$. The mobile distance is denoted as $d_{i} = \|\mathbf{p}_i^{\text{end}} - \mathbf{p}_i^{\text{start}}\|_2$, and the corresponding motion time can be given by $T^\text{mob}_i = d_{i}/{v_i}$, where $v_i$ is a decision variable subject to mechanical actuation constraints, upper-bounded by $v_i^{\max}$.

\subsubsection{Power and Energy Consumption}
To capture the energy expenditure of physical movement, we adopt a generic polynomial power consumption model that accounts for the distinct mechanical characteristics of ground-based mobile agents \cite{1638342}. The instantaneous power consumption $P^\text{mob}_i$ as a function of velocity $v_i$ is modeled as:
\begin{equation}
    P^\text{mob}_i(v_i) = P^{\text{static}}_i + \kappa_{1,i} v_i + \kappa_{2,i} v_i^2,
\end{equation}
where $P^{\text{static}}_i$ denotes the baseline power required to sustain essential onboard functionalities. The coefficients $\kappa_{1,i}$ and $\kappa_{2,i}$ represent the parameters characterizing the speed-dependent power consumption, where the linear term accounts for the dominant mechanical resistance and the quadratic term serves as an empirical approximation for internal motor losses, such as the eddy-current loss \cite{kuo1978dc}.

Consequently, the total energy consumed by agent $i$ for the mobility phase is the integral of power over time, which simplifies to:
\begin{equation}
    E^\text{mob}_i = \|\mathbf{p}_i^{\text{end}} - \mathbf{p}_i^{\text{start}}\|_2 \cdot \left( \frac{P_i^{\text{static}}}{v_i} + \kappa_{1,i} + \kappa_{2,i} v_i \right).
    \label{eq:mobility_energy}
\end{equation}


\subsection{Semantic Correlation and Computation Cost Model}\label{subsec:semantic_model}
\subsubsection{Spatio-Semantic Correlation Model}
In the proposed WAN framework, the efficiency of semantic compression is governed by the contextual overlap between the local information of an agent and its received information.

Consider the pairwise interaction in round $k$ where agent $i$ intends to transmit to agent $j$. Before this transmission, both agents fuse the data they acquired in the previous round ($k-1$). Let agent $m$ be the predecessor that transmitted to agent $i$, and agent $n$ be the predecessor that transmitted to agent $j$. We define the spatio-semantic correlation coefficient, $\rho_{u,v}^{(k)}$, for an agent $u \in \{i, j\}$ with respect to its predecessor $v \in \{m, n\}$.


First, the accumulated knowledge coverage $\mathcal{R}_u^{(k)}$ is defined as the physical union of the patrol regions associated with all information sources stored in agent $u$: $\mathcal{R}_u^{(k)} = \bigcup_{s \in \Omega_u^{(k)}} \mathcal{C}_s$.
The correlation coefficient is modeled as the Jaccard similarity index of the coverage areas:
\begin{equation}
    \rho_{u,v}^{(k)} = \frac{\text{Area}\left(\mathcal{R}_u^{(k)} \cap \mathcal{R}_v^{(k)}\right)}{\text{Area}\left(\mathcal{R}_u^{(k)} \cup \mathcal{R}_v^{(k)}\right)}, \quad \rho_{u,v}^{(k)} \in [0, 1].
\end{equation}
Note that in the initial round, $\rho_i^{(1)} = 1, \forall i \in \mathcal{N}$, since all agents only possess their own information.

Abstracting complex RAG-based deduplication via spatial overlap is mathematically essential for network optimization. Since surveillance semantics are anchored to geographic coordinates, intersecting patrol regions naturally capture redundant data. By adopting the Jaccard similarity of 2D areas as a tractable proxy, higher geometric overlap guarantees a richer shared context. This enables ELMs to synthesize compact representations, bridging network geometry with AI-driven compression efficiency.

\subsubsection{Computation Cost Model}
The system workflow dictates that the computation phase occurs simultaneously with the mobility phase. Both the sender $i$ and the receiver $j$ perform local computations to fuse and compress the data received from their respective predecessors before the new transmission $i \to j$ occurs.

Inspired by the principles of MEC \cite{8016573} and ILAC \cite{11381448}, the computational complexity (in FLOPs) for any active agent $u \in \{i, j\}$ processing data from predecessor $v \in \{m, n\}$ is modeled as:
\begin{equation}
    W^\text{comp}_u = L_{u}^{(k)} \left[ C_{\text{base}} + C_{\text{gen}} \gamma \left(1 - \rho_{u,v}^{(k)}\right) \ln \frac{1}{\eta_u} \right],
\end{equation}
where $C_{\text{base}}$ (in FLOPs/bit) is the fundamental processing cost, $C_{\text{gen}}$ is the generative cost coefficient, $\gamma$ represents the compression complexity factor, and $\eta_u$ is the semantic compression ratio of agent $u$ defined as \cite{10550151}:
\begin{equation}
    \eta_u = \frac{\text{size}(\text{Compressed Data})}{\text{size}(\text{Original Data})} = \frac{L_{u}^{(k+1)}}{L_{u}^{(k)}}.
\end{equation}

The time required for agent $u$ to complete this computation is given by $T^\text{comp}_u = \frac{W^\text{comp}_u}{f_u}$, where $f_u$ is the computational capacity of agent $u$.
The corresponding energy consumption is $E^\text{comp}_u = \tau f_u^2 W^\text{comp}_u$, where $\tau$ is the effective capacitance coefficient \cite{7762913}.

\subsection{Communication Model}\label{subsec:comm_model}
\subsubsection{Path Loss Channel Model}
Following the local computation phase, the sender agent $i$ initiates data transmission to the receiver agent $j$. The quality of the wireless channel is primarily determined by the Euclidean distance between their positions, $\mathbf{p}_i^{\text{end}}$ and $\mathbf{p}_j^{\text{end}}$. We adopt a large-scale fading model where the channel power gain $h_{ij}$ is dominated by path loss:
\begin{equation}
    h_{ij}(\mathbf{p}_i^{\text{end}}, \mathbf{p}_j^{\text{end}}) = \beta_0 \cdot \|\mathbf{p}_i^{\text{end}} - \mathbf{p}_j^{\text{end}}\|_2^{-\delta},
\end{equation}
where $\beta_0$ represents the channel gain at a reference distance of \SI{1}{m}, and $\delta \ge 2$ is the path loss exponent reflecting the propagation environment.

\subsubsection{Transmission Time and Energy}
The achievable transmission rate $R_{ij}$ (in bits/s) is given by the Shannon-Hartley theorem:
\begin{equation}\label{eq:rate_shannon}
    R_{ij} = B \log_2 \left( 1 + \frac{p_i^{\text{tx}} h_{ij}}{B N_0} \right),
\end{equation}
where $B$ is the channel bandwidth, $p_i^{\text{tx}}$ is the transmission power, and $N_0$ is the noise power spectral density at the receiver. We assume orthogonal frequency channels are assigned to active pairs to avoid interference.

Let $D_i^{\text{out}} = \eta_i L_{i}^{(k)}$ denote the size of the compressed semantic payload to be transmitted. Then, the communication time required to complete the transmission can be given by:
\begin{equation}
    T_i^{\text{comm}} = \frac{D_i^{\text{out}}}{R_{ij}} = \frac{\eta_i L_{i}^{(k)}}{B \log_2 \left( 1 + \frac{p_i^{\text{tx}} \beta_0 \|\mathbf{p}_i^{\text{end}} - \mathbf{p}_j^{\text{end}}\|_2^{-\delta}}{B N_0} \right)}.
    \label{eq:comm_time}
\end{equation}
The corresponding communication energy consumption is:
\begin{equation}
    E_i^{\text{comm}} = p_i^{\text{tx}} T_i^{\text{comm}} = \frac{p_i^{\text{tx}} \eta_i L_{i}^{(k)}}{B \log_2 \left( 1 + \frac{p_i^{\text{tx}} \beta_0 \|\mathbf{p}_i^{\text{end}} - \mathbf{p}_j^{\text{end}}\|_2^{-\delta}}{B N_0} \right)}.
    \label{eq:comm_energy}
\end{equation}

Fig.~\ref{fig.time} shows the collaborative interplay between movement, computation, and communication across two paired agents over time.

\begin{figure}[t]
    \centering
    \includegraphics[width=0.9\linewidth]{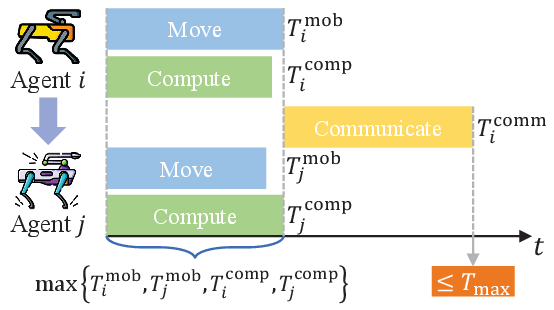}
    \caption{Joint move-compute-communicate protocol over time between two paired agents.}
    \label{fig.time}
\end{figure}

\subsection{Dynamic Topology and Aggregation Graph}
To facilitate the progressive knowledge aggregation mechanism, the network topology evolves dynamically across discrete rounds. We model the potential connections in round $k$ as a time-varying directed graph $\mathcal{G}^{(k)}=\left(\mathcal{A}^{(k)}, \mathcal{E}^{(k)}\right)$, where $\mathcal{A}^{(k)}$ denotes the set of currently active agents and $\mathcal{E}^{(k)}$ represents the set of feasible communication links between them.

\subsubsection{Active Agent Set Evolution}
Let $\mathcal{A}^{(k)} \subseteq \mathcal{N}$ denote the set of active agents at the beginning of round $k$, representing the nodes that currently hold the aggregated information state. The process initializes with all agents being active, i.e., $\mathcal{A}^{(1)} = \mathcal{N}$.
The evolution of the active set is driven by the dynamic pairing decisions. Specifically, if agent $i$ successfully transmits its accumulated payload to agent $j$ in round $k$, agent $i$ has fulfilled its aggregation duty for the current hierarchical level and enters a ``retired'' state. Conversely, agent $j$ remains active to carry the fused information for subsequent rounds. The state update rule is governed by:
\begin{equation}
    \mathcal{A}^{(k+1)} = \mathcal{A}^{(k)} \setminus \left\{ i \in \mathcal{A}^{(k)} \mid \exists j \in \mathcal{A}^{(k)}, \alpha_{i,j}^{(k)} = 1 \right\},
\end{equation}
where $\alpha_{i,j}^{(k)}$ is the binary topology decision variable defined below. The aggregation process terminates when the active set converges to a single root node (i.e., $|\mathcal{A}^{(K)}| = 1$). This final active agent serves as the dynamic aggregation root, responsible for the implicit final delivery of the global report.

\subsubsection{Topology Decision Variables}
We define the binary decision variable $\alpha_{i,j}^{(k)} \in \{0,1\}$ to indicate whether a directed communication link is established from agent $i$ to agent $j$ in round $k$:
\begin{equation}
    \alpha_{i,j}^{(k)} = 
    \begin{cases} 
    1, & \text{if agent } i \text{ transmits to agent } j \text{ in round } k, \\
    0, & \text{otherwise.}
    \end{cases}
\end{equation}
In this framework, the topology is constructed to form a valid hierarchical aggregation structure. Specifically, a converging forest that eventually evolves into a single tree rooted at the dynamically selected optimal agent. By allowing the WAN to autonomously determine the convergence path, the system can optimize the aggregation root selection based on the agents' real-time spatial and semantic states.

\section{Problem Formulation}\label{sec:problem_formulation}

The collaborative data aggregation process is mathematically formulated as a hierarchical optimization problem. This framework decomposes the complex system-wide objective into two coupled sub-problems: a inner-level optimization for pairwise interactions and a outer-level optimization for network topology evolution.

\subsection{Pairwise Interaction Optimization (Inner-Level)}
Consider a candidate pair of agents $(i, j)$ in a specific round, where agent $i$ acts as the sender and agent $j$ as the receiver. For brevity, the round index $(k)$ is omitted in this part. The total energy consumption for this interaction, denoted as $E^{\text{total}}_{(i,j)}$, is the sum of mobility, computation, and communication energy expenditures:
\begin{equation}
    E^{\text{total}}_{(i,j)} = E_i^{\text{mob}} + E_j^{\text{mob}} + E_i^{\text{comp}} + E_j^{\text{comp}} + E_i^{\text{comm}}.
    \label{eq:total_energy_pairwise}
\end{equation}

The objective is to minimize $E^{\text{total}}_{(i,j)}$ by jointly optimizing the physical movement, computational parameters, and communication resources. Let $\mathbf{x}_{ij} = \{ \mathbf{p}_i^{\text{end}}, \mathbf{p}_j^{\text{end}}, v_i, v_j, p_i^{\text{tx}}, \eta_i, \eta_j \}$ denote the set of decision variables. The pairwise optimization problem (P1) is formulated as:
\begin{subequations}
\begin{align}
    (\text{P1}): \quad \min_{\mathbf{x}_{ij}} \quad & E^{\text{total}}_{(i,j)} \tag{\theequation} \\
    \text{s.t.} \quad & \|\mathbf{p}_u^{\text{end}} - \mathbf{c}_u\|_2 \le R_u, \quad \forall u \in \{i, j\}, \label{eq:constraint_coverage} \\
    & \max \left\{ T_i^{\text{mob}}, T_j^{\text{mob}}, T_i^{\text{comp}}, T_j^{\text{comp}} \right\} + T_i^{\text{comm}} \notag \\
    & \hspace{1.6in}\le T_{\max}, \label{eq:constraint_latency} \\
    & 0 < v_u \le v_u^{\max}, \quad \forall u \in \{i, j\}, \label{eq:constraint_v} \\
    & 0 < p_i^{\text{tx}} \le P_i^{\max}, \label{eq:constraint_p} \\
    & \eta_i^{\min} \le \eta_i \le 1, \label{eq:constraint_etai} \\
    & \eta_j^{\min} \le \eta_j \le \eta_j^{\text{req}}, \label{eq:constraint_etaj}
\end{align}
\end{subequations}
where $P_i^{\max}$ denotes the maximum transmit power of agent~$i$, and $\eta_u^{\min},\forall u \in \{i, j\}$ represents the semantic compression limit. Note that the parameter $\eta_j^{\text{req}}$ specifies the requisite compression upper bound for the receiver $j$, serving as a global directive to prevent payload explosion in subsequent aggregation stages.

Regarding the constraints, \eqref{eq:constraint_coverage} enforces the mandatory geometric security coverage, restricting agents to their designated patrol regions. Constraint \eqref{eq:constraint_latency} guarantees that the interaction satisfies the strict latency deadline $T_{\max}$, explicitly capturing the protocol structure where mobility and computation occur in parallel prior to the communication phase. Finally, constraint \eqref{eq:constraint_etaj} ensures that the receiver performs sufficient semantic deduplication to maintain network scalability.

\subsection{Dynamic Topology Evolution (Outer-Level)}
Based on the optimized pairwise interaction costs derived from the inner-level analysis, the system-wide objective is to construct an energy-efficient aggregation topology. This corresponds to finding the optimal sequence of pairing decisions $\boldsymbol{\alpha} = \left\{\alpha_{i,j}^{(k)}\right\}$ across all rounds that minimizes the cumulative network energy consumption.

Let $W_{i,j}^{(k)}$ denote the minimum energy cost for the pair $(i, j)$ in round $k$, obtained by solving the inner-level problem (P1). If problem (P1) is infeasible, we set $W_{i,j}^{(k)} = \infty$ to penalize physically unrealizable links. To enforce the maximum parallelism strategy, the number of required pairings in each round is deterministic. Let $N^{(k)} = |\mathcal{A}^{(k)}|$ denote the number of active agents at the start of round $k$. The required number of links to be established is $M^{(k)} = \lfloor N^{(k)} / 2 \rfloor$. Then, the outer-level optimization problem (P2) can be formulated as:
\begin{subequations}
\begin{align}
    (\text{P2}): \quad \min_{\boldsymbol{\alpha}} \quad & \sum_{k=1}^{K} \sum_{i \in \mathcal{A}^{(k)}} \sum_{j \in \mathcal{A}^{(k)}, j \neq i} \alpha_{i,j}^{(k)} W_{i,j}^{(k)} \tag{\theequation} \\
    \text{s.t.} \hspace{1.1em} & \sum_{j \in \mathcal{A}^{(k)}, j \neq i} \alpha_{i,j}^{(k)} + \sum_{l \in \mathcal{A}^{(k)}, l \neq i} \alpha_{l,i}^{(k)} \le 1, \notag \\
    & \hspace{1.2in} \forall i \in \mathcal{A}^{(k)}, \forall k, \label{eq:role_mutex} \\
    & \sum_{i \in \mathcal{A}^{(k)}} \sum_{j \in \mathcal{A}^{(k)}, j \neq i} \alpha_{i,j}^{(k)} = M^{(k)}, \quad \forall k, \label{eq:parallelism_const} \\
    & \alpha_{i,j}^{(k)} \in \{0, 1\}, \quad \forall i, j, k, \label{eq:binary_const}
\end{align}
\end{subequations}
where $K = \lceil \log_2 N \rceil$ is the aggregation depth.
Under this formulation, the recursive update of the active population size is deterministic: $N^{(k+1)} = N^{(k)} - M^{(k)} = \lceil N^{(k)}/2 \rceil$.

Constraint \eqref{eq:role_mutex} unifies the sender and receiver restrictions. It strictly enforces that any active agent $i$ can participate in at most one interaction per round, either as a sender ($\alpha_{i,j}=1$) or as a receiver ($\alpha_{l,i}=1$), but not both. This inherently prevents directionality conflicts and multi-hop relaying within a single time slot.
Constraint \eqref{eq:parallelism_const} mandates that exactly $M^{(k)}$ pairs are formed in round $k$. This explicitly enforces the maximum parallelism strategy, ensuring the aggregation process proceeds at the fastest possible rate and concludes within the minimal depth $K$.
Finally, \eqref{eq:binary_const} imposes binary restrictions on the topology variables.

The coupled problems (P1) and (P2) constitute a hierarchical mixed-integer nonlinear programming (MINLP) problem. The outer-level problem (P2) determines the optimal discrete topology structure, while the inner-level problem (P1) optimizes the continuous physical parameters for each edge within that topology.

\section{Algorithm Design}\label{Sec:ad}
The outer-level topology $\boldsymbol{\alpha}$ depends on the inner-level weights $\{W_{i,j}\}$, while the calculation of $W_{i,j}$ requires solving a non-convex physical optimization problem. Direct global optimization is computationally intractable due to the combinatorial explosion of the topology space. Therefore, we propose a decoupled framework, termed the H-MAP algorithm.

\subsection{Joint Motion and Resource Optimization (Inner-Level)}
The problem (P1) is non-convex primarily due to the coupling of position variables $\mathbf{p}=\{\mathbf{p}_i^{\text{end}}, \mathbf{p}_j^{\text{end}}\}$ in the logarithmic term of the communication rate \eqref{eq:rate_shannon} and the convex velocity terms in the mobility energy \eqref{eq:mobility_energy}.

To address this, we employ the block coordinate descent (BCD) technique. We decompose the decision variables $\mathbf{x}_{ij}$ into two blocks: the motion block $\mathbf{m} = \{\mathbf{p}_i^{\text{end}}, \mathbf{p}_j^{\text{end}}, v_i, v_j\}$, and the resource block $\mathbf{r} = \{p_i^{\text{tx}}, \eta_i, \eta_j\}$.

\subsubsection{Optimal Resource Allocation}
Given fixed motion variables $\mathbf{m}$, problem (P1) can be reduced to:
\begin{subequations}
\begin{align}
    \min_{p_i^{\text{tx}}, \eta_i, \eta_j} \quad & E_i^{\text{comp}}(\eta_i) + E_j^{\text{comp}}(\eta_j) + E_i^{\text{comm}}( p_i^{\text{tx}},\eta_i) \tag{\theequation} \\
    \text{s.t.} \hspace{1.6em} & \eqref{eq:constraint_latency}, \eqref{eq:constraint_p}-\eqref{eq:constraint_etaj} \notag.
\end{align}
\end{subequations}

We first address the receiver's compression ratio $\eta_j$. Since the receiver's local computation energy $E_j^{\text{comp}}(\eta_j) \propto -\ln(\eta_j)$ is monotonically decreasing and a larger $\eta_j$ minimizes the processing delay $T_j^{\text{comp}}$, the optimal solution invariably lies on the upper bound. Thus, we have $\eta_j^* = \eta_j^{\text{req}}$, which fixes the receiver's processing time $T_j^{\text{proc}} = \max\left\{T_j^{\text{mob}}, T_j^{\text{comp}}(\eta_j^*)\right\}$.

To rigorously address the non-convex coupling between transmission power and time, we introduce the transmission duration $t_1$ as a decision variable in lieu of $p_i^{\text{tx}}$. Based on \eqref{eq:comm_time}, the transmission energy $E_i^{\text{comm}}$ can be equivalently expressed as a jointly convex function of $\eta_i$ and $t_1$:
\begin{equation}
    E_i^{\text{comm}}(\eta_i, t_1) = t_1 \frac{B N_0}{h_{ij}} \left( 2^{\frac{\eta_i L_i^{(k)}}{B t_1}} - 1 \right),
    \label{eq:comm_energy_func}
\end{equation}
where $h_{ij}$ denotes the channel gain with given position of both agents.

The transformation of variables requires carefully mapping the physical power constraint $0 < p_i^{\text{tx}} \le P_i^{\max}$ into the time domain. According to \eqref{eq:rate_shannon}, the maximum achievable data rate $R_{ij}^{\max}$ under maximum power is given by:
\begin{equation}
    R_{ij}^{\max} = B \log_2 \left( 1 + \frac{P_i^{\max} h_{ij}}{B N_0} \right).
\end{equation}
Consequently, for a given compressed payload size $\eta_i L_i^{(k)}$, the transmission time is physically lower-bounded by $t_1 \ge \frac{\eta_i L_i^{(k)}}{R_{ij}^{\max}}$.

Therefore, the resource allocation problem can be transformed into:
\begin{subequations}\label{pf.t1}
\begin{align}
\min_{\eta_i, t_1} \quad & \mathcal{J}(\eta_i, t_1) = E_i^{\text{comp}}(\eta_i) + E_i^{\text{comm}}(\eta_i, t_1) \tag{\theequation}\\
\text{s.t.} \quad & T_i^{\text{proc}}(\eta_i) + t_1 \le T_{\max}, \label{eq:sub_delayi}\\
& T_j^{\text{proc}} + t_1 \le T_{\max}, \label{eq:sub_delayj}\\
& t_1 \ge \frac{\eta_i L_i^{(k)}}{R_{ij}^{\max}}, \label{eq:sub_power_limit}\\
& \eta_i^{\min} \le \eta_i \le 1, \label{eq:sub_eta}
\end{align}
\end{subequations}
where $T_i^{\text{proc}}(\eta_i) = \max\left\{T_i^{\text{mob}}, T_i^{\text{comp}}(\eta_i)\right\}$. Then, we have the following theorem.
\begin{theorem}\label{theorem1}
    Problem \eqref{pf.t1} is a convex optimization problem.
\end{theorem}

\begin{IEEEproof}
The convexity of problem \eqref{pf.t1} is established by examining the convexity of the objective function and the feasible set.

First, consider the objective function $\mathcal{J}(\eta_i, t_1)$. The computation energy term is given by $E_i^{\text{comp}}(\eta_i) = \tau f_i^2 L_i^{(k)} \left[C_{\text{base}} - C_{\text{gen}}\gamma(1-\rho^{(k)})\ln \eta_i\right]$. Since the function $f(x) = -\ln(x)$ is strictly convex for $x > 0$, $E_i^{\text{comp}}(\eta_i)$ is convex with respect to $\eta_i$. The communication energy term $E_i^{\text{comm}}(\eta_i, t_1)$ in \eqref{eq:comm_energy_func} can be viewed as the perspective function of the convex function $g(x) = \frac{B N_0}{h_{ij}} (2^x - 1)$. The perspective transformation is defined as $P_g(u, v) = v g(u/v)$, which preserves joint convexity. By letting $u = \frac{L_i^{(k)}}{B}\eta_i$ and $v = t_1$, and observing that the mapping from $\eta_i$ to $u$ is linear, we conclude that $E_i^{\text{comm}}$ is jointly convex in $\eta_i$ and $t_1$. Consequently, the objective $\mathcal{J}$ is a sum of convex functions.

Next, consider the constraints. Constraint \eqref{eq:sub_delayi} involves $T_i^{\text{proc}}(\eta_i)$, which is the pointwise maximum of a constant $T_i^{\text{mob}}$ and a convex function $T_i^{\text{comp}}(\eta_i) \propto -\ln \eta_i$. Since the pointwise maximum of convex functions preserves convexity, the sublevel set defined by \eqref{eq:sub_delayi} is convex. Constraints \eqref{eq:sub_delayj}, \eqref{eq:sub_power_limit}, and \eqref{eq:sub_eta} are linear inequalities, which define convex half-spaces and polyhedrons.

Since the problem minimizes a convex objective function over the intersection of convex sets, problem \eqref{pf.t1} is strictly convex.
\end{IEEEproof}


\begin{lemma}
    The optimal transmission time is:
    \begin{equation}
        t_1^*(\eta_i) = \min \left\{ T_{\max} - T_i^{\text{proc}}(\eta_i), T_{\max} - T_j^{\text{proc}} \right\}.
        \label{eq:optimal_t1}
    \end{equation}
\end{lemma}

\begin{IEEEproof}
Analyzing the partial derivative of the objective with respect to $t_1$, we observe that:
\begin{equation}
    \frac{\partial E_i^{\text{comm}}}{\partial t_1} = \frac{B N_0}{h_{ij}} \left( 2^z - 1 - z \ln 2 \cdot 2^z \right) < 0,
\end{equation}
where $z = \frac{\eta_i L_i^{(k)}}{B t_1}$. The inequality holds because $2^z - 1 < z 2^z \ln 2$ for all $z > 0$. This monotonic decrease indicates that extending the transmission duration always reduces energy consumption. Consequently, the optimal transmission time $t_1^*$ is determined by the binding upper bound imposed by the latency constraints, provided it respects the physical power limitation.
Thus, we have solution \eqref{eq:optimal_t1}.
For the solution to be feasible, the condition $t_1^*(\eta_i) \ge \frac{\eta_i L_i^{(k)}}{R_{ij}^{\max}}$ must hold. If this condition is violated for a specific $\eta_i$, it implies that the compression is insufficient to meet the deadline even with maximum transmit power.
\end{IEEEproof}

Substituting $t_1^*(\eta_i)$ back into the objective function reduces the problem to a single-variable convex optimization over $\eta_i$.

Given the convexity and the bounded domain $[\eta_i^{\min}, 1]$, the optimal $\eta_i^*$ can be efficiently solved using a one-dimensional bisection search.
Then, the optimal transmission power can be given by:
\begin{equation}
    p_i^{\text{tx}*} = \frac{B N_0}{h_{ij}} \left( 2^\frac{\eta_i L_i^{(k)}}{B \tau_{\text{ub}}(\eta_i^*)} - 1 \right).
    \label{eq:optimal_pi}
\end{equation}

\subsubsection{Optimal Motion Planning}
Given the fixed resource allocation parameters $\mathbf{r}$, the motion planning sub-problem can be written as:
\begin{subequations}\label{prob:motion_original_final}
\begin{align}
\min_{\mathbf{p}_i^{\text{end}}, \mathbf{p}_j^{\text{end}}, v_i, v_j} & \sum_{u \in \{i,j\}} E_u^{\text{mob}}(\mathbf{p}_u^{\text{end}}, v_u) + E_i^{\text{comm}}(\mathbf{p}_i^{\text{end}}, \mathbf{p}_j^{\text{end}}) \tag{\theequation} \\
\text{s.t.} \hspace{1.7em} & \frac{\|\mathbf{p}_u^{\text{end}} - \mathbf{p}_u^{\text{start}}\|_2}{v_u} + T_i^{\text{comm}}(\mathbf{p}_i^{\text{end}}, \mathbf{p}_j^{\text{end}}) \le T_{\max},\notag \\
& \hspace{1.7in} \forall u \in \{i,j\}, \label{eq:move_time} \\
& \eqref{eq:constraint_coverage},\eqref{eq:constraint_v} \notag.
\end{align}
\end{subequations}
Problem \eqref{prob:motion_original_final} is non-convex due to the fractional structure of the mobility energy function and the complex coupling of position variables in the communication rate constraint.
To address these challenges, we propose a convex reformulation method based on variable substitution and successive convex approximation (SCA).

First, let $d_u(\mathbf{p}_u^{\text{end}}) = \|\mathbf{p}_u^{\text{end}} - \mathbf{p}_u^{\text{start}}\|_2$ denote the travel distance. To eliminate the non-convex reciprocal dependence on velocity $v_u$, we introduce an auxiliary variable $t_u^{\text{m}}$ representing the motion duration. The relationship between velocity and motion time is given by:
\begin{equation}\label{eq.v}
    v_u = \frac{d_u(\mathbf{p}_u^{\text{end}})}{t_u^{\text{m}}}.
\end{equation}
Substituting \eqref{eq.v} into \eqref{eq:mobility_energy}, the mobility energy can be reformulated as a function of the terminal position $\mathbf{p}_u^{\text{end}}$ and motion time $t_u^{\text{m}}$:
\begin{equation} \label{eq:mobility_convex_form}
    E_u^{\text{mob}}(\mathbf{p}_u^{\text{end}}, t_u^{\text{m}}) = P_u^{\text{static}} t_u^{\text{m}} + \kappa_{1,u} d_u(\mathbf{p}_u^{\text{end}}) + \kappa_{2,u} \frac{d_u^2(\mathbf{p}_u^{\text{end}})}{t_u^{\text{m}}}.
\end{equation}
Then, we have the following lemma.
\begin{lemma} \label{lemma:mobility_convexity}
The reformulated mobility energy function $E_u^{\text{mob}}(\mathbf{p}_u^{\text{end}}, t_u^{\text{m}})$ is jointly convex with respect to $\{\mathbf{p}_u^{\text{end}}, t_u^{\text{m}}\}$.
\end{lemma}

\begin{IEEEproof}
The convexity is established by analyzing the three terms in \eqref{eq:mobility_convex_form} individually.
The first term, $P_u^{\text{static}} t_u^{\text{m}}$, is affine and thus convex.
The second term is the composition of a non-negative scalar $\kappa_{1,u}$ and the Euclidean norm, which is known to be convex.
The third term can be recognized as the perspective function of the convex quadratic function $f(x) = \kappa_{2,u} x^2$.
Specifically, let
\begin{equation}
    g(\mathbf{p}_u^{\text{end}}, t_u^{\text{m}}) = t_u^{\text{m}} f\left(\frac{d_u(\mathbf{p}_u^{\text{end}})}{t_u^{\text{m}}}\right) = \kappa_{2,u} \frac{\|\mathbf{p}_u^{\text{end}} - \mathbf{p}_u^{\text{start}}\|_2^2}{t_u^{\text{m}}}.
\end{equation}
Since the perspective operation preserves convexity and the squared Euclidean norm is convex, the third term is jointly convex in $(\mathbf{p}_u^{\text{end}}, t_u^{\text{m}})$.
As the sum of convex functions is convex, the proof is complete.
\end{IEEEproof}

Next, we address the non-convexity arising from the communication rate. Let $s_{ij}$ be a slack variable representing the squared Euclidean distance, i.e., $\|\mathbf{p}_i^{\text{end}} - \mathbf{p}_j^{\text{end}}\|_2^2$.
The achievable transmission rate can be given by:
\begin{equation}
    R_{ij}(s_{ij}) = B \log_2\left(1 + \frac{\xi_{ij}}{s_{ij}^{\delta/2}}\right),
\end{equation}
where $\xi_{ij} = \frac{p_i^{\text{tx}} \beta_0}{B N_0}$ is a constant. Then, we have the following lemma.
\begin{lemma} \label{lemma:rate_convexity}
For a path loss exponent $\delta \ge 2$, the rate function $R_{ij}(s_{ij})$ is convex and monotonically decreasing with respect to $s_{ij} > 0$.
\end{lemma}

\begin{IEEEproof}
Let $f(s) = R_{ij}(s) = \frac{B}{\ln 2} \ln(1 + \xi_{ij} s^{-\theta})$, where $\theta = \delta/2 \ge 1$.
The first derivative is given by:
\begin{equation}
    f'(s) = -\frac{B \theta \xi_{ij}}{\ln 2} \cdot \frac{s^{-\theta-1}}{1 + \xi_{ij} s^{-\theta}} = -\frac{B \theta \xi_{ij}}{\ln 2} \frac{1}{s^{\theta+1} + \xi_{ij} s}.
\end{equation}
Since $\theta, \xi_{ij}, s > 0$, we have $f'(s) < 0$, confirming monotonicity.
The second derivative is:
\begin{equation}
    f''(s) = \frac{B \theta \xi_{ij}}{\ln 2} \cdot \frac{(\theta+1)s^\theta + \xi_{ij}}{(s^{\theta+1} + \xi_{ij} s)^2}.
\end{equation}
Given $\theta \ge 1$, all terms in the numerator and denominator are positive for $s>0$. Thus, $f''(s) > 0$, establishing the convexity of $R_{ij}(s_{ij})$.
\end{IEEEproof}

Although $R_{ij}(s_{ij})$ is convex, it appears in the denominator of the time calculation, creating a non-convex feasible region.
To tackle this, we employ the SCA method. At the $l$-th iteration, we approximate $R_{ij}(s_{ij})$ using its first-order Taylor expansion around the local point $\tilde{s}_{ij}^{(l)}$:
\begin{equation} \label{eq:rate_approximation}
    \tilde{R}_{ij}^{(l)}(s_{ij}) \triangleq R_{ij}(\tilde{s}_{ij}^{(l)}) - \nabla_{ij}^{(l)} (s_{ij} - \tilde{s}_{ij}^{(l)}),
\end{equation}
where the gradient coefficient $\nabla_{ij}^{(l)}$ is given by:
\begin{equation}
    \nabla_{ij}^{(l)} = -\frac{\partial R_{ij}}{\partial s_{ij}}\Bigg|_{\tilde{s}_{ij}^{(l)}} = \frac{B \xi_{ij} \delta}{(\ln 4) \tilde{s}_{ij}^{(l)} \left( (\tilde{s}_{ij}^{(l)})^{\delta/2} + \xi_{ij} \right)}.
\end{equation}
Since $R_{ij}(s_{ij})$ is convex, the affine approximation $\tilde{R}_{ij}^{(l)}(s_{ij})$ serves as a global lower bound, i.e., $\tilde{R}_{ij}^{(l)}(s_{ij}) \le R_{ij}(s_{ij})$. Replacing the original rate with this lower bound ensures that the approximate feasible set is a subset of the original one.

We define $t_2$ as the transmission time variable. The convex motion planning problem at the $l$-th iteration is formulated as:
\begin{subequations}\label{prob:motion_convex_final}
\begin{align}
    \min_{\substack{\mathbf{p}_i^{\text{end}}, \mathbf{p}_j^{\text{end}},\\ t_i^\text{m}, t_j^\text{m}, t_2, s_{ij}}} \quad & \sum_{u \in \{i,j\}} E_u^{\text{mob}}(\mathbf{p}_u^{\text{end}}, t_u^{\text{m}}) + p_i^{\text{tx}} t_2 \tag{\theequation}\\
    \text{s.t.} \quad & \frac{\eta_i L_i^{(k)}}{t_2} \le \tilde{R}_{ij}^{(l)}(s_{ij}), \label{eq:con_rate_linear} \\
    & t_u^{\text{m}} + t_2 \le T_{\max}, \quad \forall u \in \{i,j\}, \label{eq:con_total_latency} \\
    & \|\mathbf{p}_i^{\text{end}} - \mathbf{p}_j^{\text{end}}\|_2^2 \le s_{ij}, \label{eq:con_dist_slack} \\
    & \|\mathbf{p}_u^{\text{end}} - \mathbf{c}_u\|_2 \le R_u, \quad \forall u \in \{i,j\}, \label{eq:con_geo_coverage} \\
    & \|\mathbf{p}_u^{\text{end}} - \mathbf{p}_u^{\text{start}}\|_2 \le v_u^{\max} t_u^{\text{m}}, \quad \forall u \in \{i,j\}. \label{eq:con_vel_max}
\end{align}
\end{subequations}
Here, constraint \eqref{eq:con_rate_linear} can be rewritten as $t_2 \tilde{R}_{ij}^{(l)}(s_{ij}) \ge \eta_i L_i^{(k)}$. Since $\tilde{R}_{ij}^{(l)}(s_{ij})$ is affine, this is a rotated second-order cone (SOC) constraint, which is convex.
Constraint \eqref{eq:con_vel_max} is also a standard SOC constraint.
Consequently, problem \eqref{prob:motion_convex_final} is a disciplined convex programming problem efficiently solvable via standard solvers.
The SCA algorithm iteratively updates the expansion point $\tilde{s}_{ij}^{(l+1)} \leftarrow s_{ij}^*$ until convergence, guaranteeing a Karush-Kuhn-Tucker (KKT) stationary point solution.
Then, the optimized velocity can be given by \eqref{eq.v}.

Algorithm \ref{alg:inner_level} outlines the overall procedure of the proposed BCD-based framework.

\begin{algorithm}[t]
\caption{Joint Motion and Resource Optimization for Inner-Level Interaction between a Paired Agents}
\label{alg:inner_level}
\begin{algorithmic}[1]
\REQUIRE Agent indices $i, j$, round index $k$, initial states $\mathbf{p}^{\text{start}}$, payload $L^{(k)}$, and system parameters.
\STATE \textbf{Initialize:} Feasible motion variables $\mathbf{m}^{(0)} = \{\mathbf{p}_i^{\text{end}}, \mathbf{p}_j^{\text{end}}, v_i, v_j\}$. Set iteration $n=0$.
\REPEAT
    \STATE \textbf{Block 1: Optimal Resource Allocation}
    \STATE Set receiver's compression ratio $\eta_j^* = \eta_j^{req}$.
    \STATE Obtain optimal sender compression ratio $\eta_i^*$ via 1D bisection search based on the convex objective in \eqref{pf.t1}.
    \STATE Calculate optimal transmission time $t_1^*$ using \eqref{eq:optimal_t1} and recover transmit power $p_i^{\text{tx}*}$ via \eqref{eq:optimal_pi}.
    \STATE Update resource block $\mathbf{r}^{(n+1)} = \{p_i^{\text{tx}*}, \eta_i^*, \eta_j^*\}$.

    \STATE \textbf{Block 2: Optimal Motion Planning}
    \STATE Initialize SCA iteration $l=0$ and expansion point $\tilde{s}_{ij}^{(0)} = \|\mathbf{p}_i^{\text{end}} - \mathbf{p}_j^{\text{end}}\|_2^2$.
    \REPEAT
        \STATE Compute gradient $\nabla_{ij}^{(l)}$ and construct affine lower bound $\tilde{R}_{ij}^{(l)}(s_{ij})$ using \eqref{eq:rate_approximation}.
        \STATE Solve the convex problem \eqref{prob:motion_convex_final} using CVX.
        \STATE Update expansion point $\tilde{s}_{ij}^{(l+1)} \leftarrow s_{ij}^*$.
        \STATE $l \leftarrow l + 1$.
    \UNTIL{SCA objective value converges.}
    \STATE Update motion block $\mathbf{m}^{(n+1)}$ with optimized values.

    \STATE $n \leftarrow n + 1$.
\UNTIL{Total energy improvement is below threshold $\epsilon$.}
\RETURN $W_{i,j}^{(k)} = E_{(i,j)}^{\text{total}}(\mathbf{x}_{ij}^*)$ and optimized strategy $\mathbf{x}_{ij}^*$.
\end{algorithmic}
\end{algorithm}

\subsection{Dynamic Topology Evolution (Outer-Level)}
The outer-level problem (P2) aims to determine the optimal aggregation topology sequence $\boldsymbol{\alpha}$ to minimize the cumulative network energy.
Since obtaining the globally optimal solution is computationally intractable due to the combinatorial nature of the problem, we propose an efficient algorithm based on a heuristic greedy strategy.
A naive greedy approach, selecting the topology that minimizes energy consumption only for the current round, is suboptimal. This is because the pairing decision in round $k$ determines the active set $\mathcal{A}^{(k+1)}$ and their initial states for round $k+1$, including positions and accumulated data volume. A myopic decision that saves energy in the current round may leave the active agents spatially dispersed, incurring prohibitive mobility and transmission costs in subsequent stages.

To address this inter-round coupling without incurring the computational prohibitiveness of dynamic programming, we propose a potential-field guided predictive matching strategy. This approach augments the instantaneous energy cost with a potential function that estimates the future cost of the active agents.

\subsubsection{Network Centroid and Potential Field}


We postulate that an energy-efficient topology should guide the active agents to converge spatially towards the ``center of mass" of the WAN as the hierarchy climbs. 
This minimizes the expected travel and transmission distances in future rounds. 
Let $\bar{\mathbf{P}}^{(k)}_\text{c}$ denote the geometric centroid of the active agents in round $k$, defined as:
\begin{equation}\label{eq.center}
    \bar{\mathbf{P}}^{(k)}_\text{c} = \frac{1}{|\mathcal{A}^{(k)}|} \sum_{u \in \mathcal{A}^{(k)}} \mathbf{p}_u^{\text{start}}.
\end{equation}
For a candidate pair $(i, j)$ where agent $j$ is the receiver, the optimized terminal position $\mathbf{p}_j^{\text{end}}$ obtained from the inner-level problem (P1) represents its starting position for round $k+1$. We define the future potential cost $\Phi_{j}^{(k)}$ as a heuristic metric representing the projected transport burden for the active agent. It is modeled as the product of the path loss-weighted spatial deviation and the accumulated information volume:
\begin{equation}\label{eq.phi}
    \Phi_{j}^{(k)}\left(\mathbf{p}_j^{\text{end}},\eta_i,\eta_j\right) = \zeta \cdot \underbrace{\left\| \mathbf{p}_j^{\text{end}} - \bar{\mathbf{P}}^{(k)}_\text{c} \right\|_2^{\delta}}_{\text{Spatial Factor}} \cdot \underbrace{L_j^{(k+1)}\left(\eta_i,\eta_j\right)}_{\text{Data Factor}},
\end{equation}
where $\zeta$ is a weighting factor balancing current and future costs. This potential function penalizes pairing decisions that result in active agents drifting away from the network center.

\subsubsection{Augmented Weight Matrix Construction}
For every possible pair of agents $(i, j)$ in the current active set $\mathcal{A}^{(k)}$, we first solve the inner-level problem (P1). Let $W_{i,j}^{(k)}$ denote the minimized pairwise energy returned by the BCD algorithm.
We construct a directed weight matrix $\mathbf{W}^{(k)} \in \mathbb{R}^{N^{(k)} \times N^{(k)}}$, where the entry $w_{i,j}^{(k)}$ represents the augmented cost of agent $i$ transmitting to agent $j$:
\begin{equation}
    w_{i,j}^{(k)} = 
    \begin{cases} 
    W_{i,j}^{(k)} + \Phi_{j}^{(k)}\left(\mathbf{p}_j^{\text{end}*},\eta_i^*,\eta_j^*\right), & \text{if P1 is feasible}, \\
    \infty, & \text{otherwise}.
    \end{cases}
    \label{eq:augmented_weight}
\end{equation}
Note that the matrix is asymmetric since the roles of sender and receiver involve different energy profiles.

\subsubsection{Optimal Matching via Graph Theory}
With the augmented weight matrix, the topology determination for round $k$ is transformed into a minimum weight perfect matching (MWPM) problem. However, a strict perfect matching requires an even number of vertices. In our dynamic system, the number of active agents $N^{(k)}$ may be odd, which implies that one agent must remain unmatched (i.e., idle) and carry its information state directly to the next round.

To rigorously address both even and odd cases within a unified mathematical framework, we construct an augmented graph $\hat{\mathcal{G}}^{(k)} = \left(\hat{\mathcal{A}}^{(k)}, \hat{\mathcal{E}}^{(k)}\right)$. The augmented vertex set $\hat{\mathcal{A}}^{(k)}$ is defined as:
\begin{equation}
    \hat{\mathcal{A}}^{(k)} = 
    \begin{cases} 
    \mathcal{A}^{(k)}, & \text{if } N^{(k)} \text{ is even}, \\
    \mathcal{A}^{(k)} \cup \{v_{\text{virt}}\}, & \text{if } N^{(k)} \text{ is odd},
    \end{cases}
\end{equation}
where $v_{\text{virt}}$ is a virtual node acting as a placeholder for the idle agent. We extend the weight definitions such that the cost of connecting any real agent $i$ to the virtual node is zero, i.e., $w_{i, v_{\text{virt}}}^{(k)} = 0$. This ensures that the matching algorithm effectively selects the agent that is most ``costly" to pair as the idle candidate, minimizing the total operational energy of the active pairs.

The optimization problem is then formulated on this even-cardinality set $\hat{\mathcal{A}}^{(k)}$ to find the boolean matrix $\boldsymbol{\alpha}^{(k)}$:
\begin{subequations}
\begin{align}
    \min_{\boldsymbol{\alpha}^{(k)}} \quad & \sum_{i \in \hat{\mathcal{A}}^{(k)}} \sum_{j \in \hat{\mathcal{A}}^{(k)}, j \neq i} \alpha_{i,j}^{(k)} w_{i,j}^{(k)} \tag{\theequation}\\
    \text{s.t.} \hspace{1.1em} & \sum_{j \in \hat{\mathcal{A}}^{(k)}} \alpha_{i,j}^{(k)} + \sum_{l \in \hat{\mathcal{A}}^{(k)}} \alpha_{l,i}^{(k)} = 1, \quad \forall i \in \hat{\mathcal{A}}^{(k)}, \label{eq:degree_constraint}\\
    & \alpha_{i,j}^{(k)} \in \{0, 1\}.
\end{align}
\end{subequations}
Constraint \eqref{eq:degree_constraint} unifies the flow conservation rules, strictly enforcing that every active agent in the augmented graph participates in exactly one interaction per round, ensuring a valid perfect matching.
Specifically, if a pair $(i, j)$ is formed between two real agents, standard transmission occurs; if a pair $(i, v_{\text{virt}})$ is formed, agent $i$ is designated as idle, which performs no transmission in round $k$ but remains in the active set for round $k+1$, i.e., $i \in \mathcal{A}^{(k+1)}$.

Since the underlying topology allows arbitrary peer-to-peer connections, forming a general graph rather than a bipartite one, we resolve the problem by transforming the directed weights into undirected ones. The symmetrized weight $\tilde{w}_{i,j}^{(k)}$ for any pair $\{i, j\} \subseteq \mathcal{A}^{(k)}$ is given by $\tilde{w}_{i,j}^{(k)} = \min \left\{ w_{i,j}^{(k)}, w_{j,i}^{(k)} \right\}$.
Edges connected to the virtual node are naturally undirected with zero weight. We then employ the Edmonds' Blossom algorithm, which solves the MWPM problem on general graphs with the optimal solution.

The solution $\boldsymbol{\alpha}^{(k)*}$ yields the optimal pairing strategy for the current round. The set of active agents for the next round, $\mathcal{A}^{(k+1)}$, is constructed by collecting the receivers from real pairs and any agent matched with the virtual node. This process repeats recursively until a single root node remains.

\subsection{Overall H-MAP Procedure}
Rather than pursuing the globally optimal solution which is computationally prohibitive, the proposed H-MAP algorithm serves as an efficient heuristic solver. It orchestrates the inner-level optimization and the outer-level predictive matching to construct a high-quality solution. The complete execution flow, covering the recursive aggregation from $N$ agents down to a single root, is formally presented in Algorithm \ref{alg:outer_optimization}.

At each round, a central hub collects lightweight agent states to compute and broadcast the optimal pairing assignments. Subsequently, the paired agents autonomously execute the peer-to-peer move-compute-communicate protocol without further central intervention.

\begin{algorithm}[t]
\caption{Overall H-MAP Algorithm for Energy Minimization in the Considered WAN}
\label{alg:outer_optimization}
\begin{algorithmic}[1]
\REQUIRE Set of agents $\mathcal{N}$, initial positions $\mathbf{p}^{\text{start}}$, initial payloads $L^{(1)}$, and system parameters.
\STATE \textbf{Initialize:} Active set $\mathcal{A}^{(1)} = \mathcal{N}$, round index $k = 1$.
\WHILE{$|\mathcal{A}^{(k)}| > 1$}
    \STATE Calculate network centroid $\bar{\mathbf{P}}^{(k)}_\text{c}$ using \eqref{eq.center}.
    \FOR{each pair $(i, j)$ in $\mathcal{A}^{(k)}$ with $i \neq j$}
        \STATE Obtain optimal interaction cost $W_{i,j}^{(k)}$ and receiver state via Algorithm \ref{alg:inner_level}.
        \STATE Compute directed weight $w_{i,j}^{(k)} = W_{i,j}^{(k)} + \Phi_{j}^{(k)}$ if feasible; else $\infty$.
    \ENDFOR
    \STATE Form augmented graph $\hat{\mathcal{G}}^{(k)}$ by adding a virtual node $v_{\text{virt}}$ if $|\mathcal{A}^{(k)}|$ is odd.
    \STATE Compute undirected weights $\tilde{w}_{i,j}^{(k)} = \min \left\{ w_{i,j}^{(k)}, w_{j,i}^{(k)} \right\}$ for all edges.
    \STATE Find optimal matching $\mathcal{M}^*$ on $\hat{\mathcal{G}}^{(k)}$ using Edmonds' Blossom algorithm.
    \STATE $\mathcal{A}^{(k+1)} \leftarrow \emptyset$.
    \FOR{each matched pair $\{u, v\} \in \mathcal{M}^*$}
        \IF{$v_{\text{virt}} \in \{u, v\}$}
            \STATE Add the real agent $i \in \{u, v\}$ to $\mathcal{A}^{(k+1)}$.
        \ELSE
            \STATE Determine transmission direction: set $\alpha_{i,j}^{(k)} = 1$ if $w_{i,j}^{(k)} \le w_{j,i}^{(k)}$.
            \STATE Update receiver $j$'s state and add $j$ to $\mathcal{A}^{(k+1)}$.
        \ENDIF
    \ENDFOR
    \STATE Update round index $k \leftarrow k + 1$.
\ENDWHILE
\RETURN Topology sequence $\boldsymbol{\alpha}$ and corresponding $\mathbf{x}$ for each paired agent obtained from Algorithm \ref{alg:inner_level}.
\end{algorithmic}
\end{algorithm}

\subsection{Algorithm Analysis}
We analyze the computational complexity of the H-MAP algorithm by evaluating the costs of the inner-level and outer-level components hierarchically.

\textbf{Inner-Level Complexity:} For a single pair of agents, the complexity of Algorithm~\ref{alg:inner_level} is governed by the alternating BCD iterations and the specific SCA sub-iterations within each block.
\begin{itemize}
    \item The resource allocation sub-problem involves a one-dimensional bisection search over $\eta_i$, which requires $\mathcal{O}(\log(1/\epsilon_1))$ operations for a given precision $\epsilon_1$.
    \item The motion planning sub-problem is a convex SOC problem solved via an interior point method. For a problem with $V$ variables, the complexity is approximately $\mathcal{O}(V^{3.5} \log(1/\epsilon_2))$ for a given precision $\epsilon_2$ \cite{lobo1998applications}. In problem \eqref{prob:motion_convex_final}, the number of variables $V=8$.
\end{itemize}
Crucially, the problem dimension $V$ is fixed and independent of the total network size $N$. Let $I_{\text{BCD}}$ and $I_{\text{SCA}}$ denote the iteration counts for the BCD and SCA loops, respectively. The total complexity of one inner-level execution is $\mathcal{O}\left(I_{\text{BCD}} I_{\text{SCA}} V^{3.5} \log(1/\epsilon_2)\right)$. Since all terms are independent of $N$, we treat this cost as a constant factor, denoted by $C_{\text{inner}}$, in the asymptotic analysis of the network-wide algorithm.

\textbf{Outer-Level Complexity:} The outer-level evolution spans $K = \lceil \log_2 N \rceil$ rounds. The complexity consists of two parts:
\begin{itemize}
    \item Weight Matrix Construction: The system computes pairwise costs for all potential links. The number of edges is roughly $\left(N^{(k)}\right)^2$. For each edge, Algorithm \ref{alg:inner_level} is executed once. Thus, the cost is $\mathcal{O}\left((N^{(k)})^2 C_\text{inner}\right)$.
    \item Topology Matching: The Edmonds' Blossom algorithm is employed to solve the MWPM problem on a general graph with $N^{(k)}$ vertices, which has a known time complexity of $\mathcal{O}\left((N^{(k)})^3\right)$ \cite{edmonds1965paths}.
\end{itemize}
The number of active agents reduces geometrically as the hierarchy evolves, i.e., $N^{(k)} \approx N/2^{k-1}$. The total computational complexity $\mathcal{C}_{\text{total}}$ is the sum over all $K$ rounds:
\begin{equation}
    \mathcal{C}_{\text{total}} \approx \sum_{k=1}^{\lceil \log_2 N \rceil} \left[ C_{\text{inner}} \left(\frac{N}{2^{k-1}}\right)^2 + \left(\frac{N}{2^{k-1}}\right)^3 \right].
\end{equation}
This is a geometric series dominated by the first term where $k=1$. Consequently, the asymptotic complexity of the proposed H-MAP framework is $\mathcal{O}(N^3)$. This polynomial complexity ensures that the algorithm is scalable, avoiding the combinatorial explosion associated with global exhaustive search.

\section{Simulation Results}\label{Sec:sr}
In the simulations, we consider a WAN scenario where agents are randomly distributed within a square region of $500 \times 500$~\si{m^2}. The patrol radius for each agent is randomly generated from the interval $[80, 100]$~\si{m}. The initial data payload size for each agent varies between \SI{5}{Mbits} and \SI{10}{Mbits}. Unless otherwise specified, the key system parameters adopted in the simulations are summarized in Table~\ref{tb1}.

\begin{table}[t]
\renewcommand\arraystretch{1.15} 
\centering
\caption{Main System Parameters}
\begin{tabular}{|c||c|}
    \toprule\hline
    \textbf{Parameter}  & \textbf{Value} \\
    \hline
    Number of agents $N$ & $10$ \\ \hline
    Maximum latency $T_{\max}$ & \SI{6}{s} \\ \hline
    Channel Bandwidth $B$ & \SI{1}{MHz}\\ \hline
    Channel reference gain $\beta_0$ & \SI{-40}{dB} \\ \hline
    Path loss exponent $\delta$ & 3 \\ \hline
    Maximum transmit power $P_i^{\max}$ & \SI{30}{dBm} \\ \hline
    Maximum velocity $v_i^{\max}$ & \SI{5}{m/s} \\ \hline
    Computational capacity $f_i$ & \SI{1}{GHz} \\ \hline    
    Effective capacitance coefficient $\tau$ & $10^{-28}$ \\ \hline
    Weighting factor $\zeta$ & $10^{-14}$ \\
    \hline\bottomrule
\end{tabular}
\label{tb1}
\end{table}

\subsection{Inner-Level Analysis}

\begin{figure}[t]
    \centering
    \includegraphics[width=\linewidth]{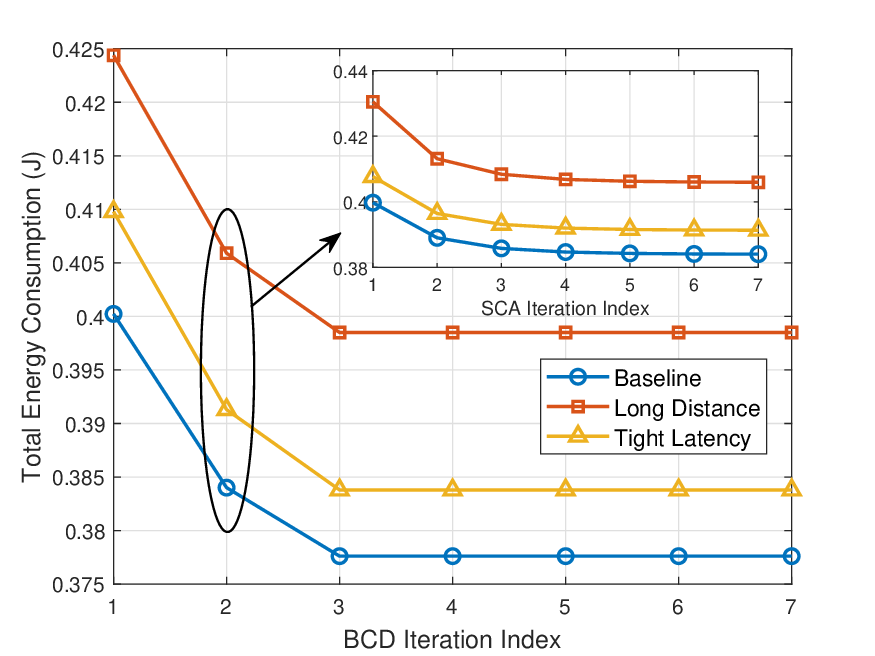}
    \caption{Convergence behavior of the proposed inner-level joint motion and resource optimization algorithm, illustrating the alternating BCD iterations and the SCA sub-iterations within the second BCD block.}
    \label{fig.convergence}
\end{figure}

Fig.~\ref{fig.convergence} illustrates the convergence behavior of the proposed joint motion and resource optimization algorithm. It can be observed that the BCD algorithm exhibits rapid convergence, typically requiring only three alternating iterations to achieve stability. Furthermore, within each BCD block, the SCA procedure also converges quickly, reaching a stable objective value in approximately four to five sub-iterations. These results demonstrate that the proposed inner-level algorithm achieves effective convergence with low computational complexity.

To evaluate the effectiveness of the proposed inner-level algorithm, we compare it against the following five benchmark schemes:
\begin{itemize}
    \item \textbf{Max Transmit Power:} This scheme assumes that each agent always transmits at its maximum power budget.
    \item \textbf{Fixed Speed:} In this scheme, each agent moves at a constant speed of $v_i^{\max}/2$ without optimized motion control.
    \item \textbf{No Semantic:} This scheme represents conventional communication where no semantic compression is applied to the data payload.
    \item \textbf{No RAG:} This scheme excludes the RAG mechanism, implying that the semantic correlation coefficient is fixed at zero.
    \item \textbf{No Motion:} This scheme assumes a static scenario where all agents remain stationary at their initial locations.
\end{itemize}







\begin{figure*}[t]
    \centering
    \vspace{-2em}
    \subfigure{
        \begin{minipage}{0.33\textwidth}
            \centering
            \includegraphics[width=1\textwidth]{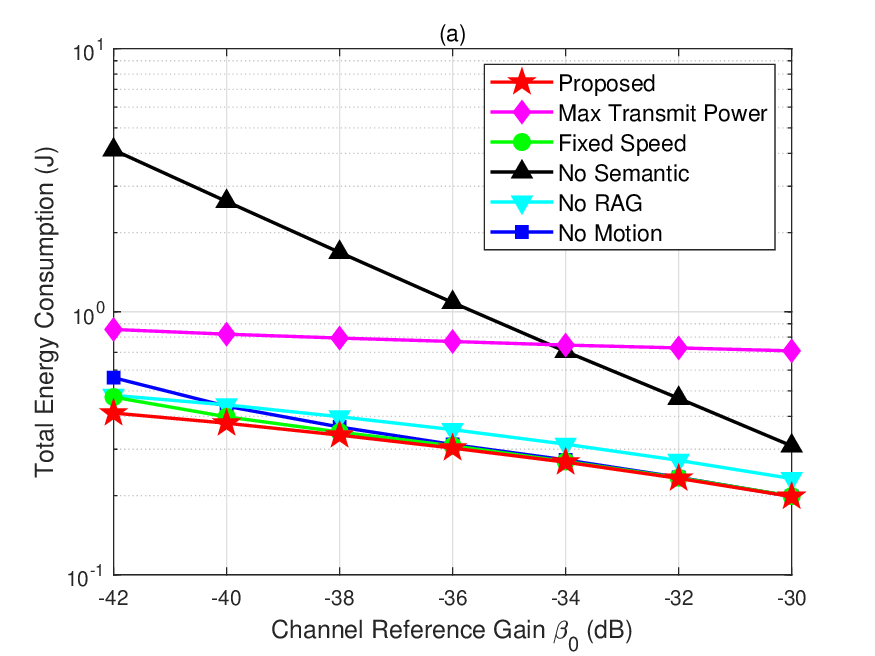}
            \label{fig.beta0}
    \end{minipage}}
    \hspace{-5mm}
    \subfigure{
        \begin{minipage}{0.33\textwidth}
            \centering
            \includegraphics[width=1\textwidth]{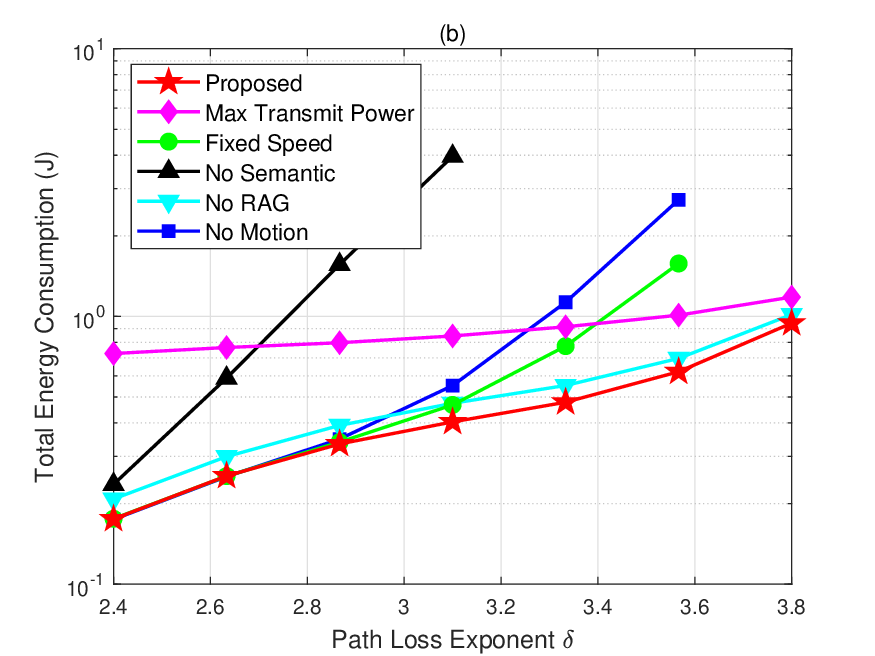}
            \label{fig.delta}	
    \end{minipage}}
        \hspace{-5mm}
        \subfigure{
            \begin{minipage}{0.33\textwidth}
                \centering
                \includegraphics[width=1\textwidth]{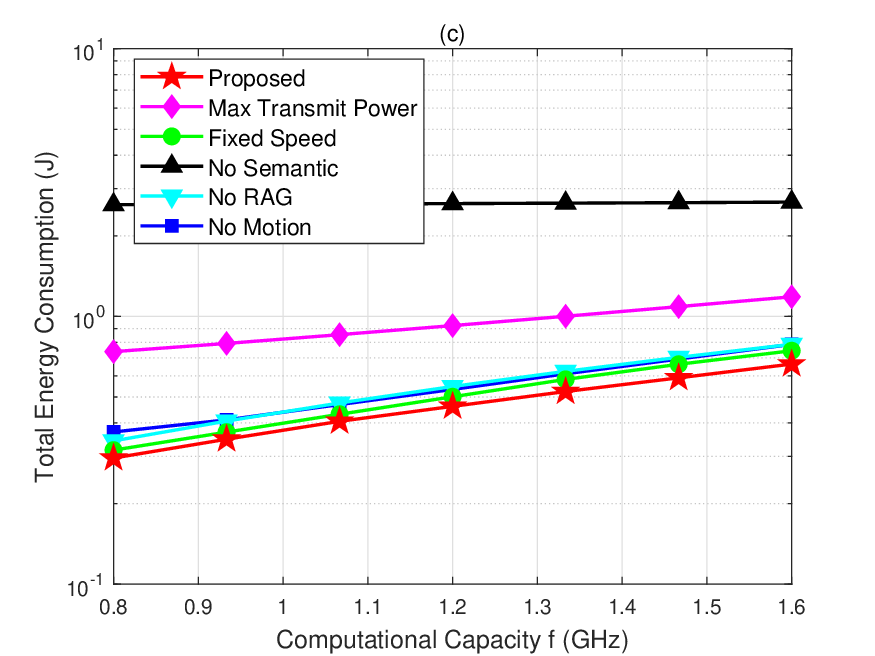}
                \label{fig.f}	
        \end{minipage}}
        \\
        \vspace{-2em}
    \subfigure{
        \begin{minipage}{0.33\textwidth}
            \centering
            \includegraphics[width=1\textwidth]{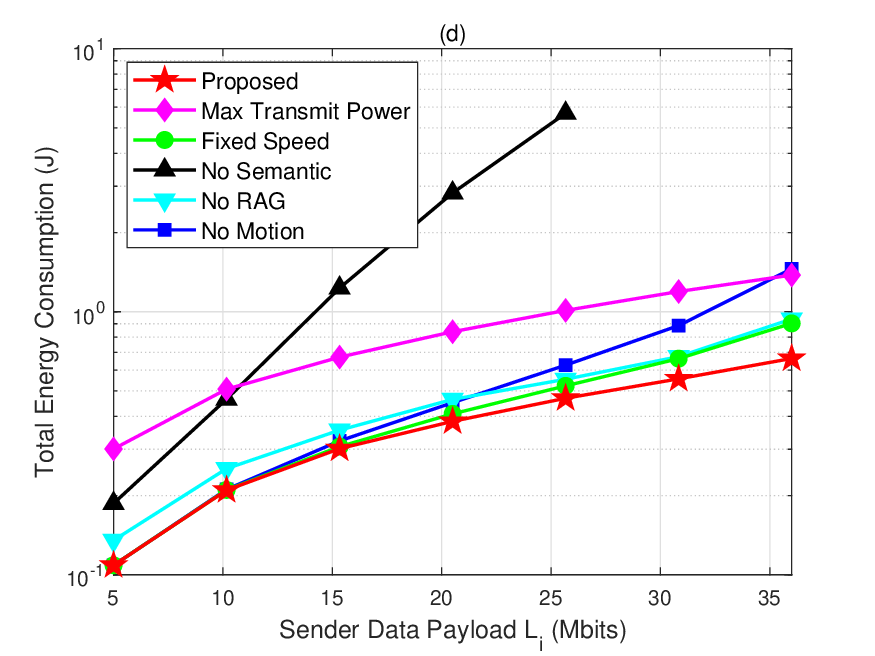}
            \label{fig.payload}
    \end{minipage}}
    \hspace{-5mm}
    \subfigure{
        \begin{minipage}{0.33\textwidth}
            \centering
            \includegraphics[width=1\textwidth]{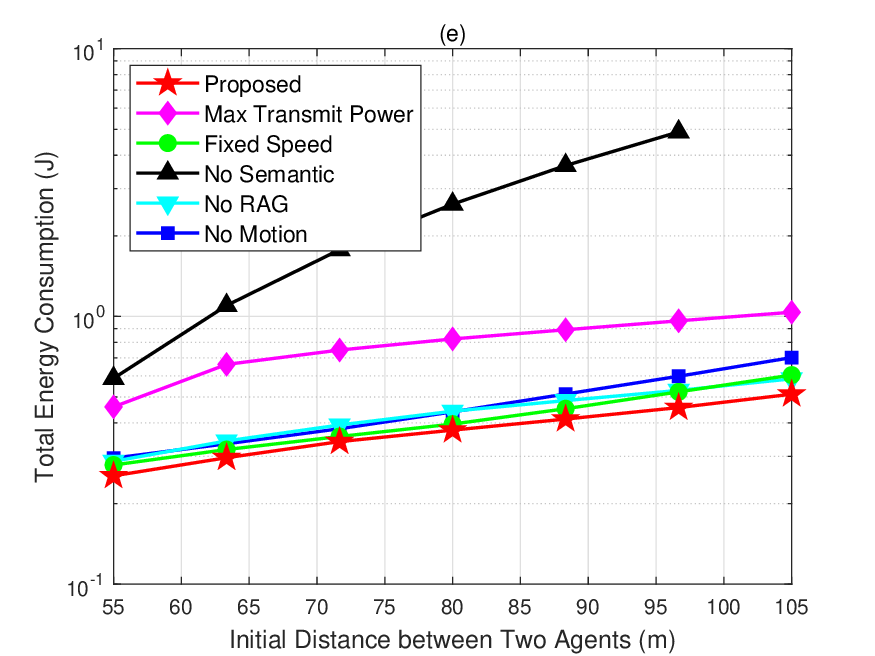}
            \label{fig.distance}	
    \end{minipage}}
        \hspace{-5mm}
        \subfigure{
            \begin{minipage}{0.33\textwidth}
                \centering
                \includegraphics[width=1\textwidth]{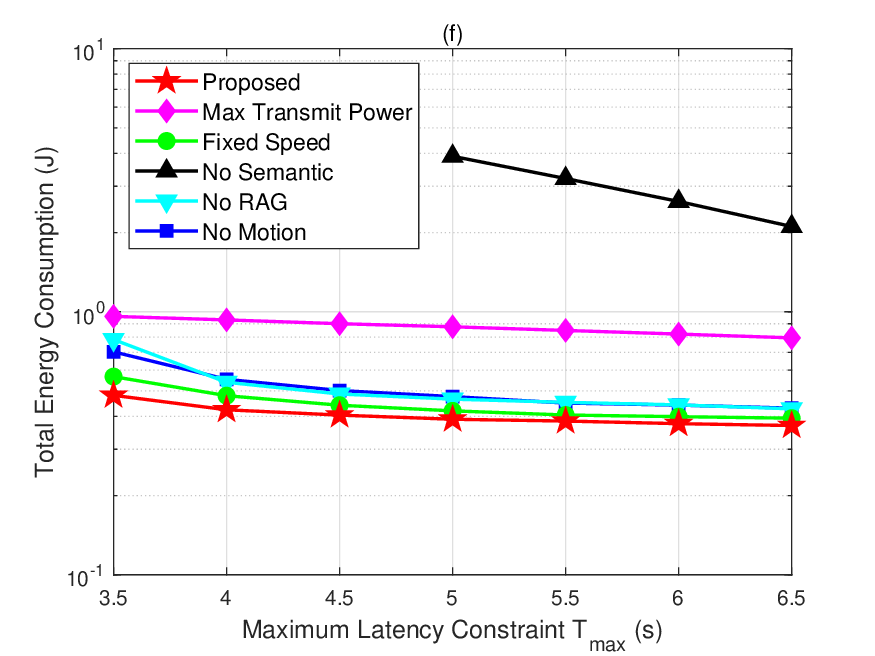}
                \label{fig.Tmax}	
        \end{minipage}}
    \vspace{-1.4em}
    \caption{Total energy consumption versus: (a) Channel reference gain $\beta_0$, (b) Path loss exponent $\delta$, (c) Computational capacity $f$, (d) Sender data payload $L_i$, (e) Initial distance between two agents, (f) Maximum latency constraint $T_{\max}$.}
    \label{fig.minor}
    \vspace{-1em}
\end{figure*}

\begin{figure}[t]
    \centering
    \includegraphics[width=\linewidth]{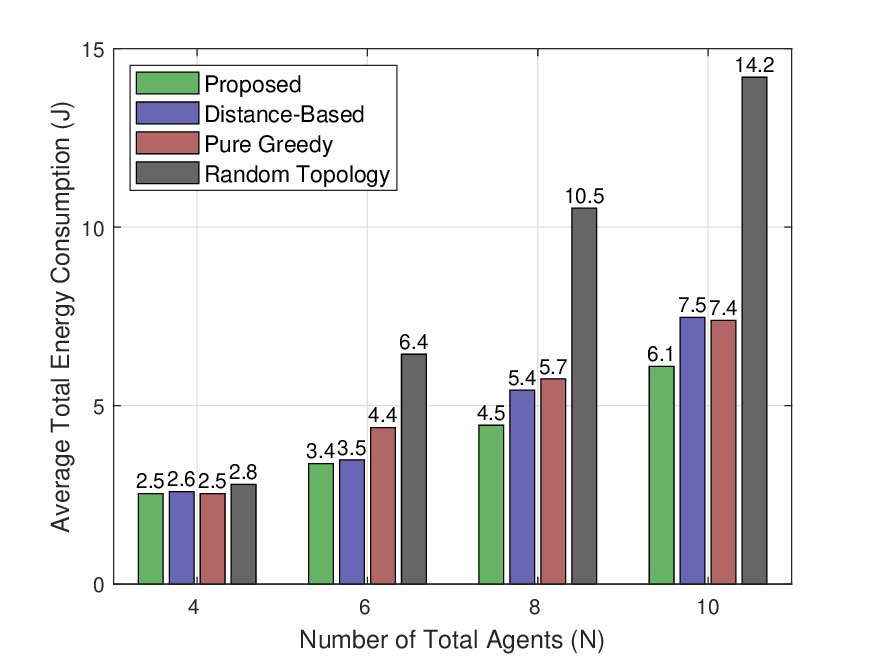}
    \caption{Average total energy consumption versus the number of total agents $N$, which is based on the average result of 10 random runs.}
    \label{fig.outer_N}
\end{figure}

\begin{figure}[t]
    \centering
    \includegraphics[width=\linewidth]{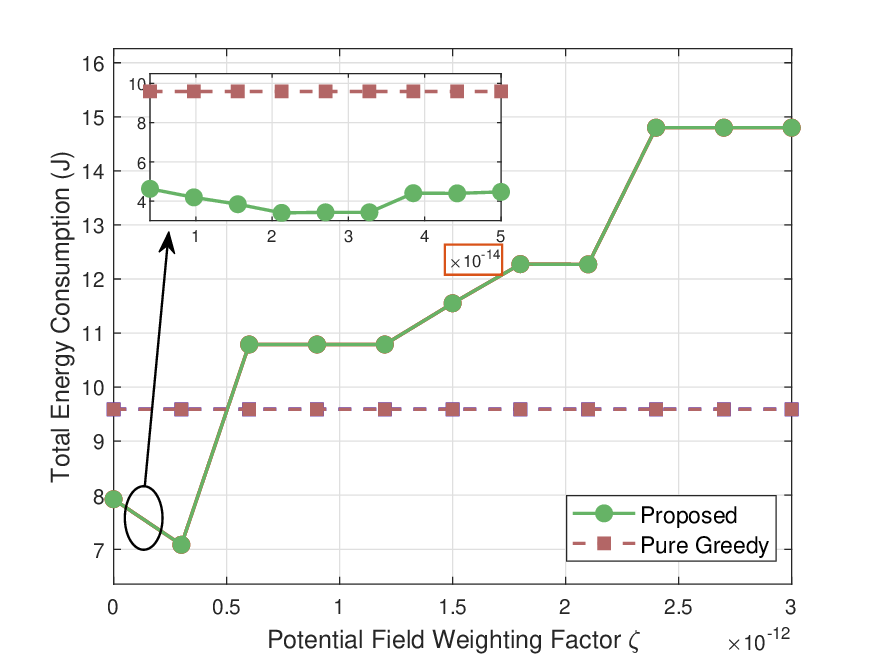}
    \caption{Impact of hyper-parameter $\zeta$ on total energy consumption.}
    \label{fig.zeta}
\end{figure}

\begin{figure*}[htbp]
    \centering
    \includegraphics[width=\linewidth]{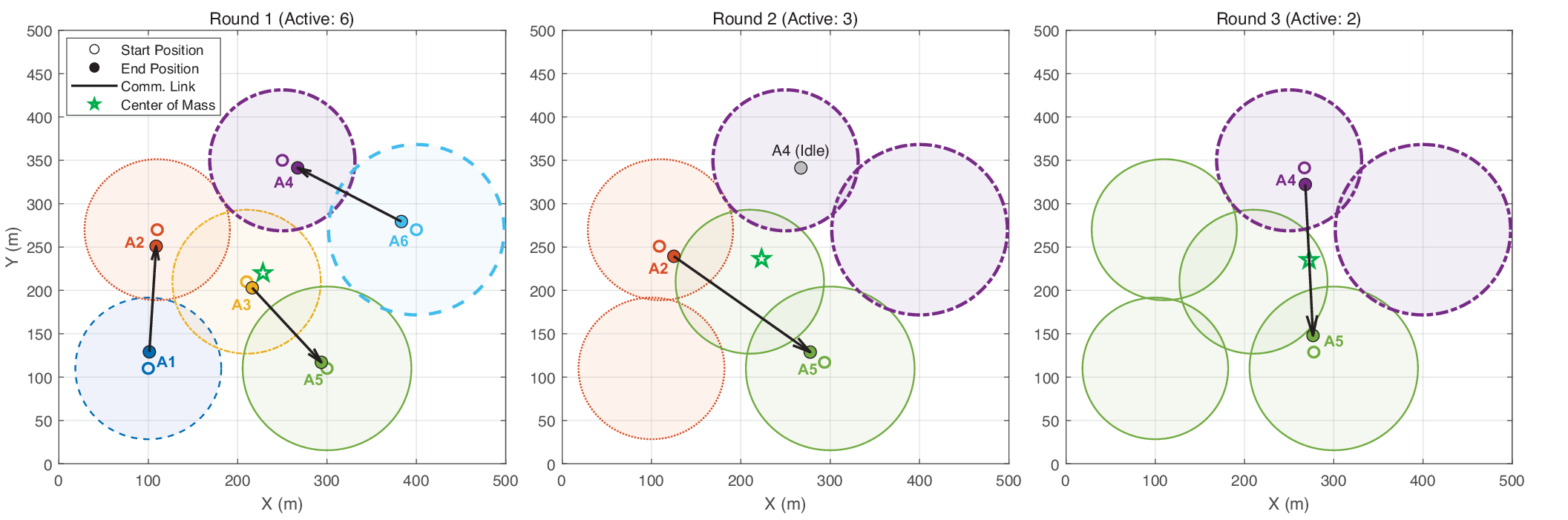}
    \caption{Visualization of the progressive knowledge aggregation process when $N=6$. The arrow signifies the directed communication link. The colored translucent circle denotes the patrol region assigned to each colored agent. When one agent transmits its information to another, its region information is fused to the received agent.}
    \label{fig.visualization}
\end{figure*}

Fig.~\ref{fig.minor} presents a comprehensive sensitivity analysis of the average total energy consumption with respect to key system parameters, including channel conditions, network topology, computational capacity, and task requirements.

\subsubsection{Impact of Channel and Topology Conditions}
Figs.~\ref{fig.beta0}, \ref{fig.delta}, and \ref{fig.distance} investigate the impact of the physical communication environment.
As illustrated in Fig.~\ref{fig.beta0}, the total energy consumption decreases monotonically as the channel reference gain $\beta_0$ increases. This is anticipated since a higher $\beta_0$ indicates superior channel quality, allowing agents to satisfy the data rate requirements with reduced transmission power.
In contrast, Figs.~\ref{fig.delta} and \ref{fig.distance} demonstrate that the energy consumption escalates with an increase in the path loss exponent $\delta$ or the initial inter-agent distance.
Crucially, it is observed that feasible solutions do not exist for ``No Semantic'' scheme when $\delta$ is excessively large or when the distance exceeds a certain threshold.

\subsubsection{Impact of Computation and Task Capabilities}
Figs.~\ref{fig.f} and \ref{fig.payload} illustrate the system energy consumption against the computational capacity $f$ and the sender data payload $L_i$, respectively. 
As shown in Fig.~\ref{fig.f}, the total energy consumption for the proposed and semantic-based benchmark schemes exhibits an upward trend as the computational capacity $f$ increases. This is due to the fact that the dynamic power consumption of the processor scales cubically with the frequency. 
Regarding Fig.~\ref{fig.payload}, the energy consumption increases with the data payload size $L_i$. This is an intuitive result, as a larger payload imposes higher demands on both the semantic processing module and the wireless transmission interface, leading to increased energy usage across all schemes.

\subsubsection{Impact of Latency Constraints}
Finally, Fig.~\ref{fig.Tmax} depicts the energy consumption versus the maximum latency constraint $T_{\max}$.
It is observed that the energy consumption decreases monotonically as the latency constraint is relaxed. This trend underscores the fundamental trade-off between delay and energy efficiency in the proposed system. Specifically, under stringent latency requirements, the agents are compelled to operate in a high-power regime to expedite both the semantic inference and data transmission processes. Conversely, a larger $T_{\max}$ alleviates these strict timing budgets, allowing the system to schedule tasks over a longer duration with significantly reduced power levels.

\subsection{Outer-Level Analysis}
To validate the scalability and effectiveness of the proposed H-MAP algorithm, we compare it against three benchmark schemes, defined as follows:
\begin{itemize}
    \item \textbf{Distance-Based:} This scheme constructs the network topology solely based on physical proximity. It assigns the Euclidean distance between agents as the edge weight and utilizes Edmonds' Blossom algorithm to find the MWPM. It ignores the semantic correlation between agents' data.
    \item \textbf{Pure Greedy:} This scheme adopts a heuristic greedy matching strategy without the guidance of the potential field. Moreover, it employs a greedy algorithm for peer selection rather than the optimal Edmonds' Blossom algorithm.
    \item \textbf{Random Topology:} This baseline randomly pairs agents to form the network topology, serving as a lower bound for performance comparison.
\end{itemize}


Fig.~\ref{fig.outer_N} illustrates the average total energy consumption versus the number of agents, $N$.
A general upward trend in energy consumption is observed for all schemes as $N$ increases. This is attributed to the fact that a larger network scale implies a greater volume of distributed sensory data that must be processed, aggregated, and transmitted to the sink node, thereby naturally accumulating higher energy costs.
However, the proposed algorithm exhibits significantly superior scalability compared to the benchmarks.
As shown in Fig.~\ref{fig.outer_N}, the performance gap between the proposed scheme and the benchmarks widens substantially as the WAN becomes denser.
For instance, at $N=10$, the proposed scheme reduces the energy consumption by approximately $19\%$ compared to the ``Distance-Based'' approach and over $57\%$ compared to the ``Random Topology'' scheme.

The superiority of the proposed framework stems from its holistic optimization strategy.
The ``Pure Greedy'' scheme suffers from a two-fold limitation: the absence of potential field guidance prevents long-term topology planning, and the use of a greedy matching algorithm fails to achieve global optimality even within a single round.
In contrast, the proposed algorithm jointly considers the semantic correlation and physical channel conditions under the guidance of the potential field, while utilizing optimal matching techniques.
This mechanism effectively aggregates agents with high semantic similarity, maximizing the redundancy elimination capability of the RAG mechanism.
Consequently, the proposed semantic-aware topology construction significantly reduces the traffic load propagated through the WAN, ensuring high energy efficiency even in large-scale scenarios.


Fig.~\ref{fig.zeta} illustrates the impact of the potential field weighting factor, $\zeta$, on the total energy consumption.
This parameter plays a crucial role in balancing the trade-off between the immediate energy expenditure and the long-term network efficiency.
For the proposed scheme, the energy consumption exhibits a convex trend.
Initially, increasing $\zeta$ reduces the total energy.
This indicates that an appropriate potential field effectively guides agents towards positions that facilitate efficient semantic aggregation in future rounds, thereby lowering the cumulative energy cost despite potential short-term overheads.
However, when $\zeta$ exceeds the optimal threshold, the energy consumption rebounds.
An excessively large $\zeta$ causes the potential field to dominate the objective function, compelling agents to undertake aggressive movements for marginal future gains.
In this regime, the immediate surge in mobility energy outweighs the long-term communication savings.
In contrast, the ``Pure Greedy'' scheme remains invariant as it operates without the potential field mechanism.


Fig.~\ref{fig.visualization} provides an intuitive visualization of the dynamic topology evolution for a scenario with $N=6$ agents.
The process exhibits a clear hierarchical structure, where the network topology is dynamically updated in rounds.
As indicated by the directed links and the fusion of patrol regions, agents progressively aggregate local semantic information and relay it toward the sink.
This layer-by-layer aggregation effectively demonstrates the progressive knowledge aggregation mechanism, where the set of active transmitting agents shrinks over time while the accumulated semantic context expands, ultimately synthesizing a comprehensive global situation report.

\section{Conclusion}\label{Sec:c}
In this paper, we have proposed a progressive knowledge aggregation framework for agentic AI-empowered WANs. Our investigation reveals that network energy efficiency is implicitly driven by spatio-semantic correlation, where RAG enables hierarchical deduplication. Furthermore, incorporating a physical potential field into topology evolution mitigates the short-sightedness of greedy matching by guiding agents toward spatially convergent states. Crucially, proactive mobility acts as a vital degree of freedom that facilitates deeper compression under strict latency constraints.

Despite these advancements, the current model assumes a simplified path loss channel and relies on centralized coordination. Future research will explore the incorporation of agentic intelligence into dynamic channel models and decentralized coordination protocols.


\bibliographystyle{IEEEtran}
\bibliography{ref}

\end{document}